\documentclass[leqno,a4paper,11pt]{article}

\usepackage{hyperref}
\usepackage{amsthm}
\usepackage{amsmath}
\usepackage{amssymb}

\usepackage{bbm}

\usepackage{mathtools}
\mathtoolsset{showonlyrefs}

\usepackage{centernot}

\newcommand{\Con}{\ensuremath{\mathcal{C}}}

\newcommand{\mb}[1]{\ensuremath{\mathbb{#1}}}
\newcommand{\N}{\mb{N}}

\newcommand{\R}{\mb{R}}

\newcommand{\G}{\ensuremath{{\mathcal G}}}

\newcommand{\sgn}{\mathop{\mathrm{sgn}}}

\newfont{\bl}{msbm10 scaled \magstep2}

\newcommand{\beq}{\begin{equation}}
\newcommand{\eeq}{\end{equation}}

\newcommand{\F}{\ensuremath{{\mathcal F}}}

\newcommand{\notmid}{\mid\kern-0.5em\not\kern0.5em}

\newcommand{\eps}{\varepsilon}

\newcommand{\supp}{\mathop{\mathrm{supp}}}

\newenvironment{pr}{\begin{proof}[\textbf{Proof:}] \ }{\end{proof}}
\newtheorem{thm}{Theorem}[section]
\newtheorem{lem}[thm]{Lemma}
\newtheorem{prop}[thm]{Proposition}

\newtheorem{cor}[thm]{Corollary}
\newtheorem{rem}[thm]{Remark}
\newtheorem{defi}[thm]{Definition}

\newcommand{\p}{\mathcal{P}}

\newcommand{\vol}{\mathrm{vol}}
\newcommand{\proj}{\mathrm{proj}}
\newcommand{\diam}{\mathrm{diam}}
\newcommand{\J}{\mathcal{J}}

\newcommand{\LLS}{Lorentzian length space }
\newcommand{\LLSn}{Lorentzian length space}
\newcommand{\LpLS}{Lorentzian pre-length space }
\newcommand{\LpLSn}{Lorentzian pre-length space}
\newcommand{\Xll}{\ensuremath{(X,d,\ll,\leq,\tau)} }

\newcommand{\nll}{\centernot\ll}

\newcommand{\dist}{\mathrm{dist}}

\renewcommand{\labelenumi}{(\roman{enumi})}
\renewcommand\theenumi\labelenumi

\newcommand{\dims}{\dim^{{\tau}}}
\newcommand{\Xllm}{\ensuremath{(X,d,m,\ll,\leq,\tau)} }
\newcommand{\cH}{{\mathcal H}}
\newcommand{\V}{\mathcal{V}}
\newcommand{\red}{\color{red}}
\newcommand{\blue}{\color{blue}}

\newcommand{\green}{\color{green}}

\usepackage{scalerel,stackengine}
\stackMath
\newcommand\reallywidehat[1]{%
\savestack{\tmpbox}{\stretchto{%
  \scaleto{%
    \scalerel*[\widthof{\ensuremath{#1}}]{\kern-.6pt\bigwedge\kern-.6pt}%
    {\rule[-\textheight/2]{1ex}{\textheight}}
  }{\textheight}%
}{0.5ex}}%
\stackon[1pt]{#1}{\tmpbox}%
}

\usepackage{tikz}
\usetikzlibrary{calc}
\usetikzlibrary{decorations.pathreplacing,angles,quotes}
\usepackage{tikz-3dplot}

\title{A Lorentzian analog for Hausdorff dimension and measure
\thanks{RM's research is supported in part by the Canada Research Chairs program and by Natural Sciences and Engineering Research Council of Canada Discovery Grants RGPIN--2015--04383 and 2020--04162. CS's research is supported by research grant J4305 of the Austrian Science Fund FWF.}
}
\author{Robert J.\ McCann\thanks{{\tt mccann@math.toronto.edu}, Department of Mathematics, University of Toronto, Canada.}\,
and Clemens S\"amann\thanks{{\tt clemens.saemann@utoronto.ca}, Department of Mathematics, University of Toronto, Canada. 
Current address: Faculty of Mathematics, University of Vienna, Austria.}}

\begin{document}

 \maketitle

 \begin{abstract}
 We define a one-parameter family of canonical volume measures on Lorentzian (pre-)length spaces. In the 
Lorentzian setting, this allows us to define a geometric dimension --- akin to the Hausdorff dimension for metric spaces --- 
that distinguishes between e.g.\ spacelike and null subspaces of Minkowski spacetime.  The volume measure corresponding to 
its geometric dimension gives a natural reference measure on a synthetic or limiting spacetime, and allows us to define what 
it means for such a spacetime to be {\em collapsed} (in analogy with metric measure geometry and the theory of Riemannian 
Ricci limit spaces).
As a crucial tool we introduce a doubling condition for causal diamonds and a notion of causal doubling measures. Moreover, 
applications to continuous spacetimes and connections to synthetic timelike curvature bounds are given.
\bigskip

\medskip


\end{abstract}
 
 \newpage
 \tableofcontents

\section{Introduction}

One of the current topics on the interface between physics and mathematics concerns the language in which General Relativity 
is formulated --- differential geometry ---  and physically relevant models that might not fit into that language. Naturally, 
as a geometric theory General Relativity is developed for smooth Lorentzian manifolds; however models of stars, for 
example, may have non-smooth matter distributions, in which case the Lorentzian metric that the Einstein equations use to 
describe such a spacetime cannot be smooth either. Moreover, fundamental to Einstein's theory of relativity is the notion of 
\emph{curvature}. In the Riemannian or metric setting analogs of sectional and Ricci curvature bounds can be defined in a 
very general setting, allowing for non-smoothness or even absence of the Riemannian metric. The common framework in the 
metric setting are \emph{length spaces} and there one can introduce curvature bounds via triangle or angle comparison leading 
to the notion of Alexandrov- or CAT$(k)$-spaces (\cite{BBI:01, BH:99}). Moreover, on a metric measure space $(X,d,m)$, i.e., 
a metric or length space with Borel measure $m$, (lower) Ricci curvature bounds are introduced using techniques from the 
theory of \emph{optimal transport} via convexity properties of entropy functionals along geodesics of probability 
measures \cite{LV:09,Stu:06a,Stu:06b, Vil:09}). This gives rise to the theory of $\mathsf{CD}(K,N)$-spaces, which 
roughly speaking, are measure spaces $(X,d,m)$ 
which behave as though they have Ricci curvature bounded below by $K$ and dimension bounded above by $N$. Such 
generalizations of curvature bounds to the \emph{synthetic} metric setting have proven extremely fruitful, even 
yielding profound insights concerning smooth Riemannian manifolds, cf.\ e.g.\ \cite{CC:97, CC:00a,CC:00b, CJN:21, Deng21}.

Inspired by this, a line of research was initiated in \cite{KS:18}, where the notion of \emph{Lorentzian length spaces} 
was introduced and shown to 
serve as a common framework for analogs of timelike and sectional curvature bounds 
while, at the same time, capturing all the relevant parts of causality theory and Lorentzian geometry. This line of research 
was extended recently in a seminal work of Cavalletti and Mondino \cite{CM:20}, where they build on ideas from the 
smooth Lorentzian setting \cite{McC:20, MS:21} to develop the concept of lower 
timelike Ricci curvature bounds in the framework of Lorentzian (pre-)length spaces. We give a brief introduction to \LLSn s in Subsection \ref{subsec-lls}.

Furthermore, spacetimes of low regularity and related topics have been the focus of a very active field of investigation in 
mathematical physics. One aim is to understand the \emph{cosmic censorship conjecture} of Penrose, cf.\ e.g.\ 
\cite{Ise:15}, which, roughly, states that the maximal globally hyperbolic development of generic initial data for the 
Einstein equations is inextendible as a suitably regular spacetime. To this end an intensive study of causality 
theory (cf.\ \cite{Min:19b}) for Lorentzian metrics of low regularity was initiated by Chru\'sciel and Grant in 
\cite{CG:12} and then pursued by various researchers, see e.g.\ \cite{Min:15,KSS:14,KSSV:14,Sae:16}. In particular, 
Chru\'sciel and Grant showed in \cite{CG:12} that for Lorentzian manifolds with merely continuous metrics pathologies in the 
causality do occur: For example, there are so-called \emph{causal bubbles}, where the boundary of the lightcone is not a 
hypersurface but has positive measure. This causal bubbling phenomenon and the failure of the so-called \emph{push-up 
property} has been investigated in \cite{GKSS:20}. A groundbreaking result for continuous Lorentzian metrics is the 
$\Con^0$-inextendibility of the Schwarzschild solution to the Einstein equations, which was shown by Sbierski in 
\cite{Sbi:18} and has initiated further research into low regularity (in-)extendibility and causality (e.g.\ \cite{GLS:18, 
GL:17, DL:17, GL:18, Sbi:21}) and in particular \cite{GKS:19}, where an inextendibility result 
for \LLSn s is given. As mentioned above, the importance of such low regularity (in-)extendibility results is rooted in 
the cosmic censorship conjecture.

Another very active field of mathematical general relativity, where low regularity features prominently in current lines 
research, is the study of \emph{singularities} and \emph{singularity theorems}, predicting causal geodesic 
incompleteness under certain causality and curvature hypotheses. The classical singularity theorems of Hawking and 
Penrose have only recently been successfully generalized to the $\Con^{1,1}$-setting (\cite{KSSV:15, KSV:15, GGKS:18}). The 
$\Con^{1,1}$-regularity class of the metric is a natural class to consider as the Riemann curvature tensor is still almost 
everywhere defined and locally bounded. Notably, these singularity theorems have recently been extended even further to the 
regularity class $\Con^1$ in \cite{Gra:20}, where the additional difficulty of non-unique geodesics had to be overcome. 
Moreover, the first singularity theorems in the synthetic setting have been proven in \cite{AGKS:21} for generalized cones, 
while an analog of the Hawking singularity theorem has been established in \cite{CM:20} using the new notion of lower timelike 
Ricci curvature bounds for \LpLSn s.

Let us also briefly discuss another natural extension of smooth Lorentzian geometry, namely cone structures on 
differentiable manifolds and Lorentz-Finsler spacetimes, see \cite{FS:12,BS:18,Min:19a, MS:19, LMO:21}, which 
provide 
new perspective on the field. 

Finally, there have been several approaches to a synthetic or axiomatic description of (parts of) Lorentzian 
geometry and causality in the past: The \emph{causal spaces} of Kronheimer and Penrose \cite{KP:67}, and the \emph{timelike 
spaces} of Busemann \cite{Bus:67}. A closely related direction of research is 
the recent approach of Sormani and Vega \cite{SV:16} and its further development by Allen and Burtscher in \cite{AB:21} of 
defining a metric on a (smooth) spacetime that is compatible with the causal structure in case the spacetime admits a time 
function satisfying an anti-Lipschitz condition. Recently, this approach has been extended to the setting of \LLSn s in 
\cite{KS:21} and it was shown that these two approaches are in a strong sense compatible.
\bigskip

One goal of our work is to 
identify any preferred choices of reference measure (and dimension) on a Lorentzian length space, 
analogous to the Hausdorff measure and dimension of a metric space.
{G}iven the centrality of these Hausdorff notions in the theory of metric spaces $(X,d)$ and Riemannian geometry, it is 
clear that one needs a satisfactory answer to develop the synthetic Lorentzian theory further. To illustrate this we review 
several results from the Riemannian or metric space case below.

Let $(M_k, g_k, p_k)_k$ be a sequence of pointed (smooth) Riemannian manifolds of the same (topological) dimension $N$ and 
Ricci curvature uniformly bounded below. Assuming that $(M_k,g_k,p_k)\to (X,d,p)$ in the Gromov-Hausdorff sense, where 
$(X,d,p)$ is a pointed length space, Cheeger and Colding \cite{CC:97} established that either
\begin{enumerate}
 \item the volumes of a balls or fixed radii collapse, i.e., $\vol^{g_k}(B_1^{d^{g_k}}(p_k))\to 0$ or
 \item they do not collapse, i.e., $\inf_k \vol^{g_k}(B_1^{d^{g_k}}(p_k))>0$.
\end{enumerate}
Thus, the notion of a \emph{collapsed} or \emph{non-collapsed} limit of Riemannian manifolds was introduced into the 
literature, with the intuition that non-collapsed limit spaces are more \emph{regular} (in a suitable sense). In particular, 
a non-collapsed limit space $(X,d)$ has Hausdorff dimension $N$, its $N$-dimensional Hausdorff measure 
is positive, i.e., $\mathcal{H}^N(X)>0$, and the volume measures $\vol^{g_k}$ converge to $\mathcal{H}^N$ in (a suitably 
adapted) weak-$*$ sense. This has been promoted to a definition in the synthetic Riemannian ($\mathsf{RCD}$) setting by De 
Philippis and Gigli in \cite{dePG:18}: An $\mathsf{RCD}(K,N)$-space $(X,d,m)$ is \emph{non-collapsed} if $m=\mathcal{H}^N$.
Here, the $\mathsf{RCD}$-condition is a strengthening of the $\mathsf{CD}$ condition, introduced by Ambrosio, Gigli and 
Savar\'e in \cite{AGS:14}, to ensure the Riemannian character of the metric measure space (e.g.\ ruling out spaces of 
Finsler type). Recently, Bru\'e and Semola established in \cite{BS:20} the following remarkably strong result: Let $(X,d,m)$ 
be a metric measure space satisfying the $\mathsf{RCD}(K,N)$-condition for $K\in\R$, $N\in(1,\infty)$. Then there is a 
$k\in\N$ with $1\leq k\leq N$ and $R\subseteq X$ (the \emph{regular set}) such that $m\rvert_{R}$ is absolutely continuous 
with respect to $\mathcal{H}^k$ and $m(X\backslash R)=0$. Yet more recently, Brena, Gigli, Honda and Zhu \cite{BGHZ:21} 
show that \emph{weakly non-collapsed} $\mathsf{RCD}$-spaces are, in fact, non-collapsed.

Consequently, our goal is to introduce to the nonsmooth Lorentzian setting a natural family of Borel measures $(\V^N)_N$ 
indexed by a dimensional parameter $N$, that play a role analogous to the family of Hausdorff measures on a metric 
space. 
Having such a family allows one to define a geometric dimension akin to the Hausdorff dimension and --- by analogy with 
the $\mathsf{RCD}$ case --- a notion of {\em non-collapsed} $\mathsf{TCD}$ 
spaces. At the same time,  our construction is novel and potentially useful even in the smooth setting.
For example,  the geometric dimension of a subspace of Minkowski space may or may not agree with its algebraic
dimension,  depending on whether or not the restriction of the metric to this subspace has a definite, semidefinite, or indefinite sign.
This provides a Lorentzian analog to the well-known fact that the Hausdorff and topological dimensions do not generally agree in subRiemannian geometries \cite{Mitchell75} \cite{GhezziJean15}. 
Our geometric measure and dimensions also provide a natural, invariant tool for quantifying the partial regularity of particular 
solutions to the Einstein equation,  analogous to the parabolic Hausdorff measure introduced by Caffarelli, Kohn and Nirenberg
to limit the size and dimension of any singular set for solutions of the Navier-Stokes equation \cite{CaffarelliKohnNirenberg82}.

 Finally, another aspect of 
our construction is its possible relation to models of quantum gravity. To be precise, a central paradigm of quantum 
gravity is that the spacetime geometry can be recovered from the causal structure and a volume element, i.e., the geometry 
and curvature is an \emph{emergent phenomenon}, see the result of Hawking, King and McCarthy \cite{HKM:76}. This is, in 
particular, the starting point of the causal set approach (\cite{BLMS:87,Sur:19}) to quantum gravity. Another approach are 
\emph{causal Fermion systems} (\cite{Fin:18}, \url{https://causal-fermion-system.com/}), see also the discussion in 
\cite[Subsec.\ 5.3]{KS:18}. Our approach remains feasible when one goes beyond Lorentzian manifolds but retains a 
notion of causal structure and time separation. Indeed, we show that one can define a one-parameter family of 
volume element(s) from just the causal structure and a notion of time separation.
\bigskip

In the following subsection we present the main results of our article and outline the structure of the paper.

\subsection{Main results and outline of the article}
The plan of the paper is as follows. In Subsection \ref{subsec-not} we introduce the relevant background on Lorentzian 
geometry and fix some notations and conventions. Then, to conclude the introduction we give a 
brief review of the theory of \LLSn s in Subsection \ref{subsec-lls}.

In Section \ref{sec-con-mea} we construct a natural family $(\V^N)_N$ of Borel measures on very general Lorentzian spaces, 
indexed by a dimensional parameter $N$, see Proposition \ref{pro-con-mea}. Then in Section \ref{sec-dim} we use this family 
to define a notion of geometric dimension $\dims$. This culminates in Corollary \ref{C:dimension}, which shows the value $N=\dims$ of
this geometric dimension to be characterized by $\V^k(X)=\infty$ and $\V^K(X)=0$ for all $0 \le k<N<K$.




Next, in Subsection \ref{subsec-len}, we investigate the relation between the length of causal curves and the 
one-dimensional Lorentzian measures of their images. In particular, images of null curves have zero geometric dimension, 
while their Hausdorff dimension is one, see Corollary \ref{cor-nul-0-dim}. Moreover, the length agrees with the 
one-dimensional measure for globally hyperbolic \LLSn s (Proposition \ref{prop-1d-mea-len}).


Section \ref{sec-dim} concludes with Subsection \ref{subsec-dim-min}, where illustrate the geometric dimension on linear 
subspaces of Minkowski spacetime. In particular, we establish that for spacelike subspaces the algebraic and geometric 
dimensions agree,  whereas for null subspaces the former exceeds the latter by one (as in the above mentioned case of null 
lines).

The main part of our article is devoted to applying the theory we develop
to spacetimes with continuous metrics \ref{sec-con-st}. 
In Subsection \ref{subsec-dou-cd} we introduce a notion of \emph{enlargement} of causal diamonds in an axiomatic way. 
Then in Subsection \ref{subsec-cyl-nhd} we establish the existence of cylindrical neighborhoods adapted to our specific 
purposes, which we then use to establish that the volume measure $\vol^g$ induced by a continuous Lorentzian metric $g$
on a spacetime $(M,g)$ of topological dimension $n$ agrees with our Lorentzian measure 
$\V^n$ in Subsection \ref{subsec-com-st}, Theorem \ref{thm-mea-vol-equ}.


After that, we introduce the notion of \emph{causally doubling measure} in Subsection \ref{subsec-cou-dou-mea} and relate it 
to the geometric dimension. In particular, a doubling constant yields a bound from above on the geometric dimension, see 
Theorem \ref{thm-dou-syn-dim}.

Finally, in Section \ref{sec-ric} we relate the synthetic timelike Ricci curvature bounds of Cavalletti and Mondino 
\cite{CM:20}, by showing for any lower bound $K \in \R$ and transport exponent $0< {\mathsf p}<1$, that --- in the 
absence of causal bubbles and timelike branching --- their weak entropic timelike curvature-dimension (respectively timelike 
measure contraction) property $\mathsf{wTCD^e_p}(K,N-1)$ (respectively $\mathsf{TMCP^e_p}(K,N)$) for a continuous globally 
hyperbolic spacetime $(M,g)$ provides a bound 
$\dim^\tau \le N$ on its geometric dimension; see Theorem \ref{T:synthetic vs geometric dimension} and Corollary 
\ref{cor-n-leq-N}. This motivates us to define:

\begin{defi}[Non-collapsed $\mathsf{TCD}$-space]
 A $\mathsf{wTCD_p^e}(K,N)$-space $X$ with reference measure $m$ is called \emph{non-collapsed} if $m=c\,\V^N$, for some 
constant $c>0$.
\end{defi}


A technical approximation result for spacetimes with continuous metrics --- used in the proof of Theorem 
\ref{T:synthetic vs geometric dimension} --- is outsourced to the Appendix \ref{app}.

\subsection{Notation and conventions}\label{subsec-not}
We fix some notation and conventions: For $\kappa\geq 0$ we denote by $\mathcal{H}^\kappa$ (and $\mathcal{H}^\kappa_\delta$), 
the $\kappa$-dimensional Hausdorff 
outer measure (at scale $\delta>0$) on a metric space $(X,d)$.

Throughout this article, $M$ denotes a smooth, connected, second countable Hausdorff manifold. We fix a smooth complete Riemannian metric $h$ on $M$ and denote the induced (length) metric by $d^h$. Unless otherwise stated, $g$ is a continuous Lorentzian metric on $M$, and we assume that $(M,g)$ is time-oriented (i.e., there exists a continuous timelike vector field $\xi$, that is, $g(\xi,\xi)<0$ everywhere). We call $(M,g)$ a \emph{continuous spacetime}. Moreover, note that our convention is that $g$ is of signature $(-+++\ldots)$. A vector $v\in TM$ is called
\begin{align}
\begin{cases}
 \text{\emph{timelike}}\\
 \text{\emph{null}}\\
 \text{\emph{causal}}\\
 \text{\emph{spacelike}}
\end{cases}
 \text{\quad if \qquad}
 g(v,v)\qquad
 \begin{cases}
  < 0\,,\\
  = 0 \text{ and } v\neq 0\,,\\
  \leq 0\text{ and } v\neq 0\,,\\
  >0\text{ or } v=0\,.
 \end{cases}
\end{align} 
Furthermore,
%
a {causal} vector $v\in TM$ is \emph{future/past directed} if $g(v,\xi)<0$ (or $g(v,\xi)>0$, respectively), where $\xi$ is 
the global timelike vector field giving the time orientation of the spacetime $(M,g)$. Analogously, one defines the causal 
character and time orientation of sufficiently smooth curves into $M$. The \emph{(Lorentzian) length} $L^g(\gamma)$ of a 
causal curve $\gamma\colon[a,b]\rightarrow M$ is defined as
\begin{equation}
 L^g(\gamma):=\int_a^b \sqrt{-g(\dot\gamma,\dot\gamma)}\,.
\end{equation}
The \emph{time separation function} is defined as follows: for $x,y\in M$ set
\begin{equation}
 \tau(x,y):=\sup\, \{ L^g(\gamma): \gamma \text{ future directed causal from } x \text{ to } y\}\cup\{0\}\,.
\end{equation}

Two events $x,y\in M$ are \emph{timelike} related if there is a future directed timelike curve from $x$ to $y$, denoted by 
$x\ll y$. Analogously, $x\leq y$ if there is a future directed causal curve from $x$ to $y$ or $x=y$. The spacetime $(M,g)$ 
is \emph{strongly causal} if for every point $p\in M$ and for every neighborhood $U$ of $p$ there is a neighborhood $V$ of 
$p$ such that $V\subseteq U$ and all causal curves with endpoints in $V$ are contained in $U$, (in which case $V$ is 
called \emph{causally convex} in $U$). Furthermore, a spacetime is \emph{causally plain} if there are no \emph{causal 
bubbles}, i.e., the boundary of the lightcone is a set of zero Lebesgue measure, cf. \cite[Def.\ 1.16]{CG:12}. In particular, 
in causally plain spacetimes the push-up principle holds: $\tau(p,q)>0$ if and only if $p\ll q$, cf.\ \cite[Thm.\ 
2.12]{GKSS:20}. This causality condition is only relevant in low regularity: every spacetime with locally Lipschitz 
continuous metric is causally plain (\cite[Cor.\ 1.17]{CG:12}), hence any smooth spacetime.

Given Lorentzian metrics $g_1$, $g_2$, we say that $g_2$ has \emph{strictly wider light cones\/} than $g_1$, denoted by $g_1\prec g_2$, if for any tangent vector $X \neq 0$, $g_1(X,X) \le 0$ implies that $g_2(X,X)<0$ (cf.~\cite[Sec.\ 3.8.2]{MS:08},~\cite[Sec.\ 1.2]{CG:12}). Thus any $g_1$-causal vector is timelike for $g_2$. For $C>0$ we denote by $\eta_C$ the (scaled) Minkowski metric on $\R^{n}$, i.e.,
\begin{equation}
 \eta_C:= -C^2\, \mathrm{d}t^2 + (\mathrm{d}x_1)^2 + \ldots + (\mathrm{d} x_{n-1})^2\,.
\end{equation}
Note that for $C > 1$ one has $\eta_{C^{-1}}\prec \eta_C$.

\subsection{A concise review of Lorentzian (pre-)length spaces}\label{subsec-lls}
We conclude the introduction by giving a brief review of \LLSn s \cite{KS:18}.
\medskip

Let $X$ be a set endowed with a preorder $\leq$ and a transitive relation $\ll$ contained in $\leq$. If $x\ll y$ 
respectively $x\le y$ we call $x$ and $y$ timelike respectively causally related. The \emph{causal/timelike future} of $x\in X$ is defined as
\begin{equation}
 I^+(x):=\{y\in X: x\ll y\}\,,\qquad \qquad J^+(x):=\{y\in X: x\leq y\}\,.
\end{equation}
Analogously, one defines the \emph{causal/timelike past} $I^-(x)/J^-(x)$ of $x\in X$ as the set of points $y\in X$ such that $y\ll x$ or $y\leq x$, respectively. Moreover, the \emph{chronological/causal diamond} with vertices $x,y\in X$ is 
\begin{equation}
 I(x,y):=I^+(x)\cap I^-(y)\,,\qquad\qquad J(x,y):=J^+(x)\cap J^-(y)\,.
\end{equation}
A subset $A\subseteq X$ is \emph{causally convex} if for all $x,y\in A$ the causal diamond $J(x,y)$ is contained in $A$, i.e., $J(x,y)\subseteq A$.

If $X$ is, in addition, equipped with a metric $d$ and a lower semicontinuous map $\tau \colon X\times X \to [0, \infty]$ that satisfies the reverse triangle 
inequality
\begin{equation}
 \tau(x,y) + \tau(y,z)\leq \tau(x,z)\,
\end{equation}
(for all $x\leq y\leq z$), as well as $\tau(x,y)=0$ if $x\nleq y$ and $\tau(x,y)>0 \Leftrightarrow x\ll y$, then \Xll is called a \emph{Lorentzian pre-length space\/} and $\tau$ is called the \emph{time separation function\/} of $X$. Note that these assumptions on \Xll yield that the \emph{push-up principle} holds, i.e., for all $x,y,z\in X$ with $x\ll y\leq z$ or $x\leq y\ll z$ one has that $x\ll z$.

A non-constant curve $\gamma \colon I\rightarrow X$ on an interval $I \subset \R$ is called 
future-directed \emph{causal} (respectively \emph{timelike}) if $\gamma$ is locally Lipschitz continuous with respect to $d$ and if for all 
$t_1,t_2\in I$ with $t_1<t_2$ we have $\gamma(t_1)\leq\gamma(t_2)$ (respectively $\gamma(t_1)\ll\gamma(t_2)$). It
is called \emph{null\/} if, in addition to being causal, no two points on the curve are related with respect 
to $\ll$. For strongly causal continuous Lorentzian metrics, this notion of causality
coincides with the usual one (\cite[Prop.\ 5.9]{KS:18}). Moreover, it should be emphasized that for a continuous spacetime $\ll$-timelike curves need not be $g$-timelike curves, see \cite[Ex.\ 2.22]{KS:18}. However, the chronological and causal futures and pasts with respect to both notions agree (by the definition of $\ll$, $\leq$ in continuous spacetimes). Thus, there is no ambiguity in denoting them by $I^\pm(x)$ and $J^\pm(x)$, respectively.

In analogy to the theory of metric length spaces, the length of a causal curve is defined via the time separation
function:  For $\gamma \colon [a,b]\rightarrow X$ future-directed causal we set
\begin{equation}
L_\tau(\gamma):=
\inf\Big\{\sum_{i=0}^{N-1} \tau(\gamma(t_i),\gamma(t_{i+1})): a=t_0<t_1<\ldots<t_N=b,\ N\in\N\Big\}.
\ 
\end{equation}
For smooth and strongly causal spacetimes $(M,g)$ this notion of length coincides with the usual one:
$L_\tau(\gamma)=L^g(\gamma)$ (\cite[Prop.\ 2.32]{KS:18}).
A future-directed causal curve $\gamma \colon [a,b]\rightarrow X$ is 
\emph{maximal\/} if it realizes the time separation, i.e., if $L_\tau(\gamma) = \tau(\gamma(a),\gamma(b))$.
\medskip

Standard causality conditions can also be imposed on Lorentzian pre-length spaces, and substantial parts of the causal ladder 
(\cite{MS:08}) continue to hold in this general setting, cf.\ \cite[Subsec.\ 3.5]{KS:18} and \cite{ACS:20}. What is needed in this work are the following notions:

A causal space $(X,\ll,\leq)$ is called \emph{causal} if the relation $\leq$ is a partial order; this rules out, e.g.\ closed timelike loops. A \LpLS \Xll is called \emph{non-totally imprisoning} if for every compact $K\subseteq X$ there is a constant $C\geq 0$ such $L^d(\gamma)\leq C$ for all causal curves $\gamma$ contained in $K$. Here $L^d$ denotes the (variational) length of a curve. A \LpLS \Xll is called \emph{strongly causal} if the metric topology of $(X,d)$ agrees with the topology
generated by the chronological diamonds $I(x,y)$ ($x,y\in X$) --- also called the \emph{Alexandrov topology}. Finally, 
a non-totally imprisoning \LpLS is called \emph{globally hyperbolic} if all causal diamonds $J(x,y)$ ($x,y\in X$) are 
compact.

We need the following auxiliary result, which has not been established yet in the setting of \LpLSn s --- only for \LLSn s, 
see \cite[Thm.\ 3.26]{KS:18} (cf.\ also \cite[Lem.\ 1-2]{KP:67} for a similar result in a slightly different setting, which 
uses the topology generated by the chronological futures and pasts $I^\pm(x)$ but called Alexandrov topology there).

\begin{lem}[Strong causality implies causality]\label{lem-str-cau-cau}
Let \Xll be a strongly causal \LpLSn, then it is causal.
\end{lem}
\begin{pr}
 Assume that \Xll is not causal, i.e., there are $x\neq y$ such that $x\leq y \leq x$. By assumption the Alexandrov topology agrees with the metric topology, hence it is Hausdorff. Consequently, there is a neighborhood $U$ of $y$ and there are $x^\pm_1,\ldots, x^\pm_N\in X$ such that $V:=I(x^-_1,x^+_1)\cap \ldots I(x^-_N,x^+_N)$ is a neighborhood of $x$ and $U\cap V=\emptyset$. However, for every $i=1,\ldots,N$ we have that $x^-_i\ll x \leq y \leq x\ll x^+_i$, so $x^-_i\ll y \ll x^+_i$ by push-up, and thus
 $y\in U\cap V$ --- a contradiction.
\end{pr}

Lorentzian length spaces are close analogues of metric length spaces in the sense that the time separation function
can be calculated from the length of causal curves connecting causally related points. A \LpLS that satisfies some 
additional technical assumptions (cf.\ \cite[Def. 3.22]{KS:18}) is called a \emph{\LLSn\/} if
for any  $x,y\in X$
\begin{equation*}
\tau(x,y)= \sup\{L_\tau(\gamma):\gamma \text{ future-directed causal from }x \text{ to } y\}\cup\{0\}\,.
\end{equation*}

Any smooth strongly causal spacetime is an example of a Lorentz\-ian length space. More generally,
spacetimes with low regularity metrics \cite[Sec.\ 5]{KS:18}, certain Lorentz-Finsler spaces \cite{Min:19a} and warped products of a line with a (metric) length space \cite{AGKS:21} provide further examples.

Finally, one can then define curvature bounds in terms of triangle comparison with the two-dimensional Lorentzian model 
spaces of constant curvature, cf.\ \cite[Sec.\ 4]{KS:18} or Ricci curvature bounds via convexity properties of entropy 
functionals on the space of probability measures on a \LpLSn, cf.\ \cite{CM:20} (and Section \ref{sec-ric}).

\section{Construction of the measures}\label{sec-con-mea}
The construction in this section works for slightly more general spaces than \LpLSn s. In fact, the time separation function $\tau$ need not be lower semicontinuous, the reverse triangle inequality and the push-up principle, i.e., $\tau(p,q)>0$ if and only if $p\ll q$, are not required to hold. Therefore, in this section we consider a set $X$, a transitive and reflexive relation $\leq$ on $X$, and $\tau\colon X\times X\rightarrow [0,\infty]$. Of course, later we consider Lorentzian (pre-)length spaces with their time separation function $\tau$.

\begin{defi}[Volume of a causal diamond]\label{def-vol-cd}
Let $N\in[0,\infty)$. Set $J(p,q):=J^+(p)\cap J^-(q)$ for $p,q\in X$ and for $\tau(p,q)<\infty$ set
 \begin{equation}\label{eq-def-rho}
   \rho_N(J(p,q)):=\omega_N\, \tau(p,q)^{N}\,,
 \end{equation}
 where $\omega_N:=\frac{\pi^{\frac{N-1}{2}}}{N\,\Gamma(\frac{N+1}{2})2^{N-1}}$ and $\Gamma(x):=\int_0^\infty t^{x-1}e^{-t}\, dt$ is Euler's gamma function.
\end{defi}

The motivation for defining $\rho_N$ in this way comes from the volumes of causal diamonds in $N$-dimensional Minkowski spacetime, when $N\geq 2$ is an integer.
Recall the Lebesgue volume of the unit ball $B_1(0) \subset \R^{N-1}$ is $\alpha_{N-1}:=\mathcal{H}^{N-1}(B_1(0)) = \pi^{\frac{N-1}{2}}/\Gamma(\frac{n+1}{2})$; the remaining factors $N\,2^{N-1}$
 in $\omega_N$ come from the volume formula for a solid cone with base $B_{\tau/2}(0) \subset \R^{N-1}$.
 
\begin{lem}[Volumes of causal diamonds in Minkowski spacetime]
 Let $N\geq 2$ be an integer and for $\tau(p,q)<\infty$ let $\tilde J(\tilde p,\tilde q)$ be a causal diamond in $N$-dimensional Minkowski spacetime $\R^{N}_1$ with $\tau(p,q)=\tilde\tau(\tilde p,\tilde q):=\tau^{\R^{N}_1}(\tilde p,\tilde q)$ then
 \begin{equation}
   \rho_N(J(p,q))=\vol^{\R^{N}_1}(\tilde J(\tilde p,\tilde q))=\omega_N\, \tau(p,q)^{N}\,.
 \end{equation}
 Moreover, the the volume $\rho_N(J(p,q))=\vol^{\R^{N}_1}(\tilde J(\tilde p,\tilde q))$ is independent of the choice of $\tilde p,\tilde q$. In particular, in case $\tau(p,q)>0$, we can find $p',q'\in\R^N_1$ with $\tilde\tau(p',q')=\tilde\tau(\tilde p,\tilde q)=\tau(p,q)$ and $p',q'$ only differ in the time component, i.e., $\proj_{\R^{N-1}}(p')=\proj_{\R^{N-1}}(q')$.
\end{lem}
\begin{pr}
For simplicity we drop the $\tilde{\,}$ notation for this proof as we only work in $N$-dimensional Minkowski spacetime $\R^{N}_1$.

In case $p\leq q$ but $p\nll q$ there is nothing to do as $J(p,q)$ is a point or the image of a null geodesic, as such it has zero $N$-dimensional Lebesgue measure.

Let $p,q\in\R^{N}_1$ with $p\ll q$  and set $T:=\tau(p,q)>0$, $v:=q-p$ which is by assumption future directed and timelike with length $T$. As the proper, orthochronous (or restricted) Lorentz group $SO^+(1,N-1)$ acts transitively on $\{w\in\R^{N}_1: w$ future directed timelike and $|w|=T\}$, there is an $A\in SO^+(1,N-1)$ such that $Av=(T,\vec{0})=:\bar q$. Then the affine linear transformation $\phi(z):= A(z-p)$ maps $p$ to $0$, $q$ to $(T,\vec{0})$ and (bijectively) $J(p,q)$ to $J(0,\bar q)$. Moreover, as the determinant of $A$ is $1$ we get
\begin{equation}
\vol(J(p,q))=\int_{J(p,q)} 1\,d \mathcal{H}^N = \int_{J(0,\bar q)} 1 |\det(A)|\, d\mathcal{H}^N=\vol(J(0,\bar q))\,. 
\end{equation}
Finally, we establish that $\vol(J(0,\bar q))=\omega_N\, T^{N}=\omega_N\, \tau(p,q)^{N}$. Denote by $D_r(z)\subseteq \R^N$ the closed Euclidean ball of radius $r$ around $z$. Let $v=(t,\bar x)\in \R^N_1$, then $v\in J(0,\bar q)$ if and only if $0\leq t\leq T$ and $\|\bar x\|_e\leq \min(t, T-t)$, where $\|.\|_e$ denotes the Euclidean norm on $\R^N$. Thus we get by using the known volume of (closed) Euclidean balls in $\R^N$ that
\begin{align}
 \mathcal{H}^N(J(0,\bar q))&=\int_0^T \mathcal{H}^{N-1}(D_{\min(t,T-t)}(0))\, \mathrm{d} t\\
 &= \int_{0}^{\frac{T}{2}}  \mathcal{H}^{N-1}(D_{t}(0))\, \mathrm{d} t + \int_{\frac{T}{2}}^{T}  \mathcal{H}^{N-1}(D_{T-t}(0))\, \mathrm{d} t\\
 &= 2 \frac{\pi^{\frac{N-1}{2}}}{\Gamma({\frac{N+1}{2}})} \int_0^{\frac{T}{2}} t^{N-1}\, \mathrm{d}t = \frac{\pi^{\frac{N-1}{2}}}{N \Gamma({\frac{N+1}{2}}) 2^{N-1}}\, T^N\,. \\
\end{align}
\end{pr}

\begin{defi}[Lorentzian measures from coverings by closed diamonds]
\label{defi-can-mea}
 Let $(X,\leq,\tau)$ be as above, $d$ a metric on $X$ and set $\mathcal{J}:=\{J(p,q): p,q\in X$ with $p<q\}\cup \{\emptyset\}$. Extending \eqref{eq-def-rho} by setting $\rho_N(\emptyset):=0$ and $\rho_N(J(p,q)):=\infty$ if $\tau(p,q)=\infty$
 yields a map $\rho_N\colon\mathcal{J}\rightarrow [0,\infty]$. For $\delta>0$ and $A\subseteq X$ set 
 \begin{equation}
  \V^N_\delta(A):=\inf\{\sum_{i=0}^\infty \rho_N(A_i): A_i\in\mathcal{J}, \diam(A_i)\leq \delta\ \forall i\in\N\text{ and } A\subseteq \bigcup_{i=0}^\infty A_i\}\,,
 \end{equation}
with the convention that $\inf\emptyset = \infty$. Moreover, for $A\subseteq X$ we set
\begin{equation}
 \V^N(A):=\sup_{\delta>0}\V^N_\delta(A)\,.
\end{equation}
\end{defi}

Classical results (cf.\ \cite[Thm.\ 9.3]{Els:18}, \cite[2.10.1]{Fed:69}) give:

\begin{prop}[Outer measures and measure of each given dimension]
\label{pro-con-mea}
 For $N\in[0,\infty)$ and $A\subseteq X$ the map $(0,\infty) \ni \delta \mapsto\V^N_\delta(A)$ is monotonically 
nonincreasing, $\V^N_\delta$ is an outer measure and $\V^N(A)=\lim_{\delta\searrow0}\V^N(A)$ defines a Borel measure, 
i.e. a measure $\V^N$ on the Borel subsets of $X$. 
\end{prop}

\begin{rem}[Finite versus countable coverings]
 In the definition of the Lorentzian outer measures $\V^N_\delta$ one can also restrict to finite coverings of the set to be measured. A precise argument is given in the proof of Proposition \ref{prop-1d-mea-len}.
\end{rem}

Thus, in particular, we have constructed a 
dimensionally parameterized family of measures on a \LpLS $\Xll$. In the following section \ref{sec-dim} we show that for each $N\in[0,\infty)$, analogous to Hausdorff measure in metric spaces, this can be used to define a notion of dimension for such spaces. Moreover, this measure (with $N$ equal the spacetime dimension) coincides with the volume measure in a spacetime, as will be established in Subsection \ref{subsec-com-st}.

\begin{rem}[Dependence of $\V^N$ on the metric $d$]
 A priori the measure $\V^N$ depends on the metric $d$ on $X$. However, if two metrics $d,\tilde d$ are strongly equivalent, i.e., there are constants $c_\pm>0$ such that for all $x,y\in X$: $c_-\, d(x,y)\leq \tilde d(x,y)\leq c_+\, d(x,y)$, then the measure $\V^N$, and the measure $\tilde\V^N$ (constructed with respect to $\tilde d$), agree. This can be weakened if $X$ is a \emph{Radon space}, i.e., a topological space such that every Borel measure is (inner) regular (or tight), hence for example if $X$ is a Polish topological space. In this case it suffices to have for every $x\in X$ constants $c_\pm^x>0$ such that for all $y\in X$ one has $c_-^x\, d(x,y)\leq \tilde d(x,y)\leq c_+^x\, d(x,y)$. This condition is 
 sufficient (but not necessary) for the topological equivalence of $d$ with $\tilde d$.
  Furthermore, we will see in Theorem \ref{thm-mea-vol-equ} that in the case of spacetimes it does not depend at all on the chosen background Riemannian metric as it coincides with the volume measure of the Lorentzian metric.
\end{rem}

\section{A geometric dimension for \LpLSn s}\label{sec-dim}

At this point we introduce a geometric dimension of \LpLSn s and arbitrary subsets, analogously to the Hausdorff dimension defined for metric spaces using the Hausdorff measures. This is also useful for defining hypersurfaces and appropriate measures on them.
\begin{defi}[Geometric dimension]
 If \Xll is a \LpLSn, $(\V^N)_{N\in[0,\infty)}$ the family of Borel measures constructed in Theorem \ref{pro-con-mea}, and $B\subseteq X$, 
 then the \emph{geometric dimension} $\dims(B)$ of $B$ is defined as
 \begin{equation}
  \dims(B):=\inf\{N\geq 0: \V^N(B)<\infty\}\,,
 \end{equation}
with the convention that $\inf\emptyset = +\infty$.
\end{defi}

To give a characterization of the geometric dimension we need the following notion.
\begin{defi}[Local $d$-uniformity]
 A \LpLS \Xll is called \emph{locally $d$-uniform} if every point has a neighborhood $A$ such that $\tau(x,y) = o(1)$ on $A$ as $d(x,y)\searrow 0$.
\end{defi}


\begin{lem}[Uniqueness of dimension with nontrivial measure]
\label{lem-dim-mM}
 Let \Xll be a \LpLSn, $A \subseteq B\subseteq X$ and $ 0 \le k < \dims(B) < K <\infty$. Then
 \begin{enumerate}
  \item  \label{lem-dim-mM-sub} $\dims(A)\leq\dims(B)$,
  \item  \label{lem-dim-mM-m} $\V^{k}(B)=\infty$, and
  \item  \label{lem-dim-mM-M} if \Xll is locally $d$-uniform then $\V^K(B)=0$.
 \end{enumerate}
\end{lem}
\begin{pr}
\begin{enumerate}
 \item Let $\eps>0$ then there is an $N\in [\dims(B),\dims(B)+\eps)$ such that $\V^N(B)<\infty$, hence $\V^N(A)\leq\V^N(B)<\infty$ and so $\dims(A)\leq N<\dims(B)+\eps$. As $\eps>0$ was arbitrary we get the claim.
 \item This holds by the definition of $\dims(B)$.
 \item Let $B\subseteq X$, and cover it by a countable family of neighborhoods $U_i \subset A_i$, where $\tau = o(1)$ on $A_i$ (as $d \searrow 0$) and 
  $\eta_i:=\dist(U_i,X\backslash A_i)>0$. Then $\V^K(B)\leq \sum_i\V^K(B\cap U_i)$, thus it suffices to show $\V^K(B\cap U_i)=0$.
 
 Let $\dims(B)<K'<K$ such that $\V^{K'}(B)=:G<\infty$, so $\V^{K'}(B\cap U_i)\leq G$.
 Let $0<\delta<\frac{\eta_i}{3}$ and let $(F_j)_j$ be a covering of $B\cap U_i$ by causal diamonds with $\diam(F_j)<\delta$ for all $j\in\N$ and $\sum_j\rho^{K'}(F_j)<2G$. Moreover, without loss of generality we can assume that $F_j\cap (B\cap U_i)\neq \emptyset$ for all $j\in\N$ and thus $F_j\subseteq A_i$. Therefore, $\tau=o(1)$ holds on $F_j=J(p_j,q_j)$. Finally, we estimate
 \begin{align}
  \V^K_\delta(B\cap U_i)&\leq \sum_j \rho^K(F_j)\\
  &=\frac{\omega_K}{\omega_{K'}}\sum_j\rho^{K'}(F_j)\tau(p_j,q_j)^{K-K'}\leq \frac{\omega_K}{\omega_{K'}}2 G\ o(1)^{K-K'}\to 0\,,
 \end{align}
as $\delta\searrow 0$ (as then $\diam(F_j)\searrow 0$). Thus $\V^K(B\cap U_i)=0$ and hence $\V^K(B)=0$ as claimed.
\end{enumerate}
\end{pr}

The above results allow a characterization of the geometric dimension analogous to the one for the Hausdorff dimension of a metric space, cf.\ \cite[Thm.\ 1.7.16]{BBI:01}.

\begin{cor}[Equivalent characterizations of dimension]\label{C:dimension}
  Let \Xll be a locally $d$-uniform \LpLSn. Then $N=\dims(X)$ if and only if $\V^k(X)=\infty$ and $\V^K(X)=0$ 
    for all $0 \le k<N<K <\infty$. 
  Moreover, 
  \begin{equation}
  \dims(X)= {\sup} \{L\geq 0: \V^L(X)= \infty\}\,.
  \end{equation}
\end{cor}

\begin{lem}\label{lem-dim-cup}
 Let \Xll be a locally $d$-uniform \LpLS and $X=\bigcup_{i\in\N} U_i$. Then
 \begin{equation}
  \dims(X)=\sup_{i\in\N}\dims(U_i)\,.
 \end{equation}
\end{lem}
\begin{pr}
 By Lemma \ref{lem-dim-mM},\ref{lem-dim-mM-sub} we get that $\dims(U_i)\leq\dims(X)$ for all $i\in\N$. For the converse inequality, let $K>\sup_i \dims(U_i)$, then by Lemma \ref{lem-dim-mM},\ref{lem-dim-mM-M} we know that $\V^K(U_i)=0$ for all $i\in\N$. Consequently, $\V^K(X)\leq \sum_i\V^K(U_i)=0$ and thus $\dims(X)\leq K$. As this holds for all $K>\sup_i \dims(U_i)$ we conclude that $\dims(X)=\sup_i\dims(U_i)$. 
\end{pr}


We will see in Proposition \ref{prop-con-st-d-uni-dim} that a strongly causal, causally plain
continuous spacetime is locally $d$-uniform and its geometric dimension agrees with the manifold dimension.

\subsection{One-dimensional measure versus length}\label{subsec-len}
In this subsection we investigate the relationship of $\V^1(\gamma([a,b]))$ and $L_\tau(\gamma)$, where $\gamma\colon[a,b]\rightarrow X$ is a causal curve. Note that the normalization constant in that case is $\omega_1=1$, cf. Definition \ref{def-vol-cd}.
\bigskip

The following lemma is the Lorentzian analog of \cite[Lem.\ 4.4.1]{AT:04} (note the reversal of the inequality as compared to the metric case).
\begin{lem}[Simple upper bound by the time separation]\label{lem-Nd-mea-tau}
 Let \Xll be a strongly causal \LpLSn. Let $\gamma\colon[a,b]\rightarrow X$ be a future directed causal curve and $N\in[1,\infty)$. Then
 \begin{equation}\label{eq-Nd-mea-tau}
  \V^N(\gamma([a,b]))\leq \omega_N\, \tau(\gamma(a),\gamma(b))^N\,.
 \end{equation}
\end{lem}
\begin{pr}
 Set $\Gamma:=\gamma([a,b])$, which is a compact subset of $X$. Let $\delta>0$ and set $B_t=B^d_{\delta/2}(\gamma(t))$, for $t\in [a,b]$. By the definition of strong causality there is for every $t\in[a,b]$ a causally convex open neighborhood $U_t$ of $\gamma(t)$ that is contained in $B_t$ (use a neighborhood $I(x_1,y_1)\cap\ldots\cap I(x_k,y_k)$).
 Consequently, we can cover $\Gamma$ by finitely many of them, i.e., $\Gamma\subseteq \bigcup_{j=0}^K U_j$, where $U_j:=U_{t_j}$. By connectedness of $\Gamma$ there is a partition $a= s_0 < s_1 < \ldots = s_L = b$ such that $\gamma(s_i),\gamma(s_{i+1})\in U_{j_i}$ for $i=0,\ldots, L-1$ and $j_i\in\{0,\ldots, K\}$ (cf.\ step 2 of the proof of \cite[Lem.\ 2.6.1]{BBI:01}). Setting $J_i:=J(\gamma(s_i),\gamma(s_{i+1}))$ yields that $J_i\subseteq U_{j_i}$ and so $\diam(J_i)\leq \diam(U_{j_i})\leq \delta$. Clearly, $\Gamma\subseteq\bigcup_{i=0}^{L-1} J_i$ and therefore, this gives that
 \begin{align}
  \V^N_\delta(\Gamma)&\leq \omega_N \sum_{i=0}^{L-1} \tau(\gamma(s_i),\gamma(s_{i+1}))^N 
  \\&\leq \omega_N \Bigl(\sum_{i=0}^{L-1} \tau(\gamma(s_i),\gamma(s_{i+1}))\Bigr)^N\\
  &\leq \omega_N\, \tau(\gamma(a),\gamma(b))^N\,,
 \end{align}
where we used $N\geq 1$ in the second inequality and the reverse triangle inequality in the last one. This holds for all $\delta>0$, so the claim follows.
\end{pr}

\begin{cor}[Null curves are zero-dimensional]\label{cor-nul-0-dim}
 Let \Xll be a strongly causal \LpLSn. Let $\gamma\colon[a,b]\rightarrow X$ be a future directed null curve. Then $\dims(\gamma([a,b]))=0$.
\end{cor}
\begin{pr}
 Set $\Gamma:=\gamma([a,b])$. Let $N\in(0,\infty)$ and let $\delta>0$. As in the proof of Lemma \ref{lem-Nd-mea-tau} above we construct a covering $J_i:=J(\gamma(t_i),\gamma(t_{i+1}))$ of $\Gamma$, where $a \leq t_0 < t_1 < \ldots< t_K\leq b$ is a partition of $[a,b]$ with $\diam(J_i)\leq \delta$ for all $i=0,\ldots, K-1$. Thus, as $\gamma$ null, we obtain
  \begin{align}
  \V^N_\delta(\Gamma)\leq \omega_N \sum_{i=0}^{K-1} \tau(\gamma(t_i),\gamma(t_{i+1}))^N = 0\,.
  \end{align}
  This gives that $\V^N(\Gamma)=0$ and so $\dims(\Gamma)\leq N$ for all $N\in (0,\infty)$, hence $\dims(\Gamma)=0$.
\end{pr}

So null curves have zero geometric dimension whereas their Hausdorff dimension is one (as they are Lipschitz and injective, if suitably parametrized).
\medskip

Furthermore, Lemma \ref{lem-Nd-mea-tau} above allows us to establish a Lorentzian analog of \cite[Thm.\ 4.4.2]{AT:04}.

 
 
 \begin{prop}[One-dimensional measure versus length]\label{prop-1d-mea-len}
  Let \Xll be a strongly causal \LpLSn. Let $\gamma\colon[a,b]\rightarrow X$ be a future directed causal curve. Then
 \begin{equation}
  \V^1(\gamma([a,b]))\leq L_\tau(\gamma)\,.
 \end{equation}
 Furthermore, if all causal diamonds $J(x,y)$ in $X$ are closed (as e.g.\ if $X$ is globally hyperbolic), then the length of a causal curve agrees with the one-dimensional measure $\V^1$ of its image, i.e., $\V^1(\gamma([a,b]))= L_\tau(\gamma)$.
\end{prop}
\begin{pr}
First, we establish that $\V^1(\Gamma)\leq L_\tau(\gamma)$, where $\Gamma:=\gamma([a,b])$.
 Let $a=t_0<t_1<\ldots < t_N=b$ be a partition of $[a,b]$, then by Equation \eqref{eq-Nd-mea-tau} (with $N=1$) we get that
 \begin{equation}
  \V^1(\gamma([a,b]) \leq \sum_{i=0}^{N-1} \V^1(\gamma([t_i,t_{i+1}])) \leq \sum_{i=0}^{N-1} \tau(\gamma(t_i),\gamma(t_{i+1}))\,.
 \end{equation}
Now taking the infimum over all partitions of $[a,b]$ the claim follows.

Second, we show the reverse inequality under the assumption that all causal diamonds are closed. To this end we establish that for the Lorentzian (outer) measures $\V^N_\delta$, $\V^N$ we can restrict to finite coverings of the set in question (rather than countable ones). To be precise: For $N\in[0,\infty)$, $\delta>0$ and $A\subseteq X$ set
 \begin{align}
  \tilde\V^N_\delta(A)&:=\inf\{\sum_{i=0}^k \rho_N(J_i): k\in\N, J(x_i,y_i)=J_i\in\mathcal{J}, \diam(J_i)\leq \delta\ \forall i\\
  &\qquad\quad\text{ and } A\subseteq \bigcup_{i=0}^k J_i\}\,,\\
   \tilde\V^N(A)&:=\sup_{\delta>0}\tilde\V^N_\delta(A)\,.
 \end{align}
Clearly, one has that $\V^N_\delta\leq \tilde\V^N_\delta$ and $\V^N\leq \tilde\V^N$ for all $N\in[0,\infty)$, $\delta>0$. For $N\in[0,\infty)$, $\delta>0$, and $A\subseteq X$ let $J\in\mathcal{J}$ with $\diam(J)\leq \delta$. As $J$ covers $A\cap J$ we have that $\tilde\V^N_\delta(A\cap J)\leq \rho_N(J)$. Then \cite[Thm.\ 2.10.17,(i)]{Fed:69} gives that $\tilde\V^N_\delta(A)\leq\V^N_\delta(A)\leq \V^N(A)$. Taking the supremum over all $\delta>0$ yields that $\tilde\V^N\leq\V^N$ and so $\tilde\V^N=\V^N$, hence we can restrict to finite coverings from now on.

Without loss of generality we can assume that $\gamma$ is parametrized such that it is never-locally-constant, i.e., there is no non-trivial interval $[a',b']$ in $[a,b]$ such that $\gamma\rvert_{[a',b']}$ is constant (cf.\ \cite[Exc.\ 2.5.3]{BBI:01}), as a reparametrization does not change the length, cf.\ \cite[Lem.\ 2.28]{KS:18}. Given such a parametrization of $\gamma$, the curve $\gamma$ is injective: Note that \Xll is causal by Lemma \ref{lem-str-cau-cau}. Thus if there were $a\leq s<t\leq b$ such that $\gamma(s)=\gamma(t)$, then for all $s\leq r \leq t$ one has that $\gamma(s)\leq \gamma(r)\leq \gamma(t)=\gamma(s)$, hence $\gamma(r)=\gamma(s)$ and $\gamma$ is constant on $[s,t]$ --- a contradiction.

Let $\delta>0$, and for every $\eps>0$ there is a finite covering $(J_i)_{i=0}^k$ of $\Gamma$ with causal diamonds $J_i=J(x_i,y_i)$ of diameter less or equal than $\delta$ such that
\begin{equation}\label{eq-cov-1d-mea}
 \sum_{i=0}^k \tau(x_i,y_i) < \V^1_\delta(\Gamma) + \eps\,.
\end{equation}
Moreover, we can assume that $J_i\cap \Gamma\neq\emptyset$ for all $i$. As $\Gamma$ is connected and all $J_i$s are closed we can find a finite chain $J_{i_0},\ldots, J_{i_l}$ such that $\gamma(a)\in J_{i_0}$, $\gamma(b)\in J_{i_l}$ and $\Gamma \cap J_{i_j}\cap J_{i_{j+1}}\neq\emptyset$ for all $i=0,\ldots,l-1$. For $i=1,\ldots,l-1$ choose $\gamma(t_i)\in J_{i_j}\cap J_{i_{j+1}}$. Then $a=:t_0 < t_1 < \ldots < t_{l-1} < t_l:=b$ is a partition of $[a,b]$ as $\gamma$ is injective.

Then for all $j=0,\ldots, l-1$ one has that $\tau(\gamma(t_j),\gamma(t_{j+1})) \leq \tau(x_{i_j},\gamma(t_j)) + \tau(\gamma(t_j),\gamma(t_{j+1})) + \tau(\gamma(t_{j+1}),y_{i_j})\leq \tau(x_{i_j},y_{i_j})$ as $\gamma(t_j),\gamma(t_{j+1})\in J_{i_j}=J(x_{i_j},y_{i_j})$ and by the reverse triangle inequality. Finally, by using Equation \eqref{eq-cov-1d-mea} this gives that
\begin{align}
 L_\tau(\gamma)\leq \sum_{j=0}^{l-1} \tau(\gamma(t_j),\gamma(t_{j+1})) \leq \sum_{j=0}^{l-1} \tau(x_{i_j},y_{i_j}) \leq \sum_{i=0}^k \tau(x_i,y_i) < \V^1_\delta(\Gamma)+\eps\,.
\end{align}
As this holds for all $\eps>0$ we can let $\delta\searrow 0$ to obtain that $L_\tau(\gamma)\leq \V^1(\Gamma)$, which concludes the proof.
\end{pr}

\begin{prop}[Countable sets are zero dimensional and measured by their cardinality
]\label{prop-cou-mea}
 Let \Xll be a strongly causal \LpLSn. Let $N\in[0,\infty)$ and assume additionally that in case $N>0$ we have that for all $x\in X$ and all neighborhoods $U$ of $x$ there are $x^\pm\in U$ such that $x^- < x < x^+$ and $x^-\not\ll x \not\ll x^+$. Let $A\subseteq X$ countable, then for $N>0$ we have $\V^N(A)=0$. For $A\subseteq X$ arbitrary we have $\V^0(A)=|A|$, the cardinality of $A$.
\end{prop}
\begin{pr}
Let $N\in[0,\infty)$ and assume first that $A$ is countable. Thus we can write $A$ as the countable disjoint union of its singletons, so $\V^N(A)=\sum_{a\in A} \V^N(\{a\})$. Thus it remains to show that $\V^N(\{a\})=0$ for $N>0$ and $\V^0(\{a\})=1$. Let $a\in A$ and for $\delta>0$ set $B_a:=B^d_{\delta/2}(a)$. Then by by strong causality there is a causally convex neighborhood $U_a$ of $a$ contained in $B_a$. Moreover, for $N>0$ we have by assumption that there are $a^\pm\in U_a$ such that $a^- < a < a^+$ and $a^-\not\ll a \not\ll a^+$. In case $N=0$, we can find (by strong causality) 
that there are $a^\pm\in U_a$ with $a^-\ll a \ll a^+$. Thus, in both cases, $J_a:=J(a^-,a^+)\subseteq U_a$ and $\diam(J_a)\leq \delta$. Consequently,
 \begin{equation}\label{eq-mea-sin}
  \V^N_\delta(\{a\})\leq \omega_N\, \tau(a^-,a^+)^N = \begin{cases}
                                                     0 \qquad N>0\,,\\
                                                     1 \qquad N=0\,.
                                                    \end{cases}
 \end{equation}
Note that actually equality holds in \eqref{eq-mea-sin} and as this holds for all $\delta>0$ we obtain $\V^N(A)=0$ for $N>0$ and $\V^0(A)=|A|$ as claimed. Finally, if $A\subseteq X$ is uncountable, the above shows that $\V^0(A)\geq k$ for all $k\in\N$ and so $\V^0(A)=\infty = |A|$.
\end{pr}

\begin{rem}
 The assumption in case $N>0$ of the Proposition \ref{prop-cou-mea} above is satisfied if, e.g., there is a null curve through every point in the \LpLSn. 
\end{rem}

\subsection{Dimension and measure of Minkowski subspaces}\label{subsec-dim-min}

To give some intuition for the geometric notions of dimension and measure that we have introduced, we examine them on linear subspaces of Minkowski spacetime.  
From Proposition \ref{prop-1d-mea-len} it is already clear that nontrivial subspaces on which the Lorentzian metric has negative definite restriction
have geometric dimension $1$ and that $\V^1$ agrees 
with the Lebesgue measure given by proper time on them.  The next lemma shows that the geometric dimension of a spacelike subspace also agrees with
its topological and algebraic dimensions,  and that the corresponding nontrivial Lorentzian measure is a positive multiple of Lebesgue measure on this subspace.
It is followed by an example which shows that the geometric dimension of a null subspace is one less than its topological and algebraic dimension. Here {\em null} subspace refers to a subspace on which the metric has nonnegative semidefinite,  but not positive definite, restriction.

Note that Hausdorff measure on Minkowski spacetime is not canonical,
unless one specifies a choice of Euclidean metric.  However,  
Hausdorff dimension 
is canonical,   since the Hausdorff measures of a given dimension associated to 
different Euclidean metrics are mutually
absolutely continuous with respect to each other.

\begin{lem}[Lorentzian and Hausdorff measures on spacelike subspaces]\label{L:spacelike subspace}
The restriction of $\V^k$ to a 
spacelike subspace of the Minkowski spacetime $\R^n_1$ having algebraic dimension $k$ 
is a positive multiple of the Hausdorff measure on the same subspace.
\end{lem}

\begin{proof}
By Lorentz invariance,  it is enough to establish the lemma for the canonical $k$-dimensional subspace $S=(0,\ldots,0) \times \R^k \subset \R^n_1$.
Hausdorff measure $\cH^k$ is defined using the associated Euclidean product metric on $\R^n_1$.
Since the restriction of $\V^k$ to $S$ is translation invariant,  the result follows as soon as we bound $\V^k$ above and below by constant multiples $c_\pm$ of $\cH^k$ on $S$.

Let $I=[-1,1]$ and $B \subset \R^k$ denote the ball of radius $r_k=\sqrt{k}$ circumscribed around $I^k$, 
and recall $\cH^k[B]=\alpha_k r_k^k$ where $\alpha_{k}:= \pi^{k/2}/\Gamma(\frac{k+2}{2})$ is the volume of the unit 
ball in $\R^k$.
Set $Q=(0,\ldots, 0) \times I^k \subset S$. 
  For each $\delta=2r_k/j>0$, to obtain the easy bound $\V^k_\delta[Q] \le 2^k \frac{\omega_k}{\alpha_k} \cH^k(B)$ 
just divide the cube $Q$ into $j^k$ subcubes of length $2/j$ with midpoints $(0,\ldots,0,m_i)$  with $m_i \in \R^k$ and cover
$Q$ by the $j^k$ diamonds 
$$
J((-r_k/j,0,\ldots,0,m_i),(r_k/j,0,\ldots,0,m_i))\,,
$$  
each of which contributes $\omega_k(2{r_k}/j)^k$ to an upper bound $\V^k_\delta(Q) \le {2^k \frac{\omega_k}{\alpha_k}} 
\cH^k(B)$.

To get a lower bound,   let $Q \subset \cup J_i$ denote any countable cover of $Q$ by
closed causal diamonds $J_i = J(A_i^-,A_i^+)$ of diameter at most $\delta$. The Euclidean areas $\cH^k(Q_i)$ of the intersections $Q_i=Q \cap J_i$ 
sum to at least $2^k$.  If we can find $c_->0$ such that $\omega_k\tau(A_i^-,A_i^+)^k \ge c_- {\mathcal H}^k(Q_i)$ for each $i$,  then 
since the cover was arbitrary,  
$\V^k[Q] \in [c_-2^k,  2^k \frac{\omega_k}{\alpha_k} \cH^k(B)] \subset (0,\infty)$ will follow, so that $Q$ and hence $S$ have Lorentzian dimension $k$. 

Fixing $i$, we may assume the intersection of $J^+(A_i^-)$ with $S$ is a closed $k$-ball $B^-$ of positive radius,  since otherwise there is nothing to prove.
Notice the subgroup of $SO(n,1)$ which fixes $S$ (and 
$S \cup \{A_i^+-A_i^-\}$) has dimension 
$\frac{(n-k)(n-k-1)}{2}$ (and $\frac{(n-k-1)(n-k-2)}{2}$ respectively).
The difference $n-k-1$ indexes the dimension of the set of alternate locations $A^\pm$
of $A_i^\pm$ whose time separation from each point of $S$ is the same as that of $A_i^\pm$, so that e.g.~$J^\mp(A^\pm)\cap S=B^\pm$.  
One such choice yields $A^\pm=(\pm t_\pm,a_\pm,b_\pm)$ where $t_\pm>0$ and $a_+=a_-\in \R^{n-k-1}$ (hence if 
$k=n-1$, then there are no $a^\pm$), so that
$B^\pm$ is a Euclidean $k$-ball centered at $b_\pm$ with radius $r_\pm = (t_\pm^2 - |a_-|^2)^{1/2}$.  This choice gives
$\cH^k(B^\pm) \le \alpha_{k} r_\pm^k$  and
$$
\tau(A^-,A^+)^2 = (t_++t_-)^2 - |b_+ - b_-|^2\,.
$$
It remains to estimate $\cH^k(B^+ \cap B^-) \lesssim \tau(A^-,A^+)^k$. Observe that the desired estimate is straightforward when
$B^+ \cap B^-=B^+$ or $B^+ \cap B^-=B^-$.  

In the remaining cases by dilating, we may assume $|b_+-b_-|=1$ and 
estimate $\cH^k(B^+ \cap B^-) \le  (r_++r_--1) \alpha_{k-1} r^{k-1}$ where $B^+ \cap B^-$ is contained in a right circular cylinder of height $r_+ + r_--1$ and radius $r$ 
satisfying the Pythagorean laws
\begin{align*}
r_+^2 &= r^2 + (1-s)^2\,,
\\r_-^2 &= r^2 + s^2\,,
\end{align*} 
see Figure \ref{fig-int-bal}.

Solving for $(r,s)$ yields
\begin{align*}
s &= \frac12 (1 + r_-^2 - r_+^2)\ {\rm and}
\\r^2 &= (r_--s)(r_-+s)\,.
\end{align*} 
From $-r_+ \le s -1 \le  r_+$ and $r_\pm \le t_\pm$ we find
\begin{align*}
\tau(A^+,A^-)^{2k} &=  (t_++t_--1)^k (t_++t_- +1)^k
\\ &\ge (t_++t_--1)^{k+1} (t_++t_- +1)^{k-1}
\\ & \ge (t_++t_--1)^2(t_- -s)^{k-1} (t_- +s)^{k-1}
\\ & \ge (r_++r_--1)^2r^{2k-2}
\\ &\ge \alpha_{k-1}^{-2} \cH^k(B^+ \cap B^-)^2\,,
\end{align*} 
as desired; i.e., we may take $c_- = {\omega_k}/{\alpha_{k-1}} = \frac{1}{k 2^{k-1}}$.
\end{proof}

\begin{figure}
 \begin{tikzpicture}

    \node at (-3,4) [label=above: $\R^k$] {};
    
    \node[circle,draw, minimum size=200] (c1) at (0,0){};
    \draw [fill] (c1.center) circle [radius=0.06];
    \node at (c1.center) [label=left: $b^+$] {};
    \draw [dashed] (c1.center) -- (c1.135) {};
    \node at ($ (c1.center)!.75!(c1.135) $) [label=below: $r_+$]{};
    \node at (c1.225) [left=4pt, below=4pt]{$B^+$};
    
    \node[circle,draw, minimum size=200] (c2) at (5,0){};
    \draw [fill] (c2.center) circle [radius=0.06];
    \node at (c2.center) [label=right: $b^-$] {};
    \draw [dashed] (c2.center) -- (c2.45) {};
    \node at ($ (c2.center)!.75!(c2.45) $) [label=below: $\ r_-$]{};
    \node at (c2.315) [right=4pt, below=4pt]{$B^-$};

    \draw[red] (c1.center) -- (c2.center);
    \draw[blue] (c1.center) -- (c1.45) -- (c2.center);
    \draw[blue] (c1.45) -- (intersection of c1.45--c1.315 and c1.center--c2.center);
    \node at ($ (c1.45)!0.5!(intersection of c1.45--c1.315 and c1.center--c2.center) $) [text=blue, right=1pt]{$r$};
    \draw[decoration={brace,mirror,raise=5pt},decorate, red]  (c1.center) -- node[font=\small, text=red, below=6pt] 
{$1-s$\ \ \ \ } ($ (c1.center)!0.49!(c2.center) $);
\draw[decoration={brace,mirror,raise=5pt},decorate, red] ($ (c1.center)!0.51!(c2.center) $) -- node[font=\small, 
text=red, below=6pt] {$s$} (c2.center);

    \draw[green] ($ (c1.135)!0.795!(c1.45) $) -- ($ (c1.45)!0.205!(c2.45) $) -- ($ (c2.225)!0.205!(c2.315) $) -- ($ 
(c1.225)!0.795!(c1.315) $) -- ($ (c1.135)!0.795!(c1.45) $);

\end{tikzpicture}
\caption{A schematic drawing of the (green) cylinder of height $r_+ + r_- - 1$ and radius $r$ containing 
$B^+\cap B^-\subseteq\R^k$ from the proof of Lemma \ref{L:spacelike subspace}.}\label{fig-int-bal}
\end{figure}
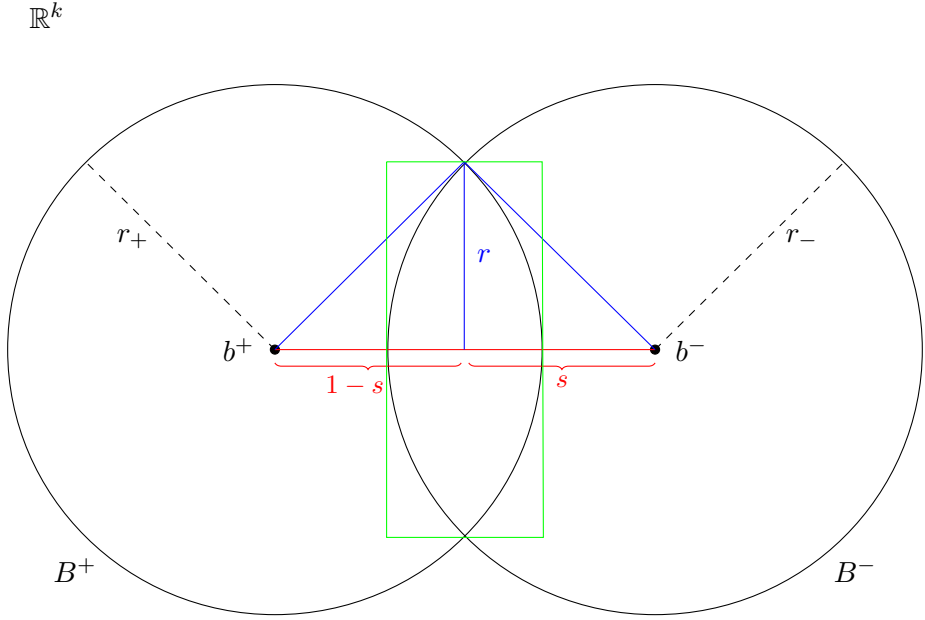

\begin{lem}[Linear null hypersurfaces have {geometric codimension two}]\label{L:null subspaces}
Let $n\geq 2$.  Let $S \subset \R^n_1$ be a null subspace of algebraic dimension $k \ne n$ in Minkowski spacetime. 
Then $\dims(S) =k-1$ and the Lorentzian measure splits as  $\V^{k-1}=c\,\mathcal{H}^{k-1}\times \mathcal{H}^0$ on 
$S=R \times \R \nu$, where $R$ is any spacelike subspace of $S$,  $\nu\in S$ is a null vector and $c=c(R,\nu)>0$.
\end{lem}

\begin{proof}
The case $k=1$ follows from Corollary \ref{cor-nul-0-dim} and Proposition \ref{prop-cou-mea}, 
thus let $k \ge 2$. Moreover, $S$ contains 
a spacelike subspace of algebraic dimension $k-1$ as well as 
a null-vector $\nu \ne 0$. Lemmas  \ref{lem-dim-mM}(i) and \ref{L:spacelike subspace} imply $\dims(S) \ge k-1$;  the 
remainder of the argument shows the opposite inequality.

Fixing a Euclidean metric on $\R^n_1$, we may suppose $\nu=(e_1-e_2)/\sqrt{2}$ and $e_1$ both have Euclidean length $1$ without loss of generality.  Now $S$ can be tiled by translates of a Euclidean unit cube $Q\subset S$ centered at the origin and having one of its faces orthogonal to $\nu$ in the Euclidean sense.  In view of Lemma \ref{lem-dim-cup}, it remains only to show $\dims(Q) \le k-1$.  

Now observe that the intersection of the causal cone $J^-(\frac12 (\delta \nu + t e_1))$ with $S$ is a paraboloid of revolution having Euclidean focal length $\sim t$.
Taking $t \ll \delta$, the intersection of the causal diamond $J_{\delta,t}:=J(-\frac12(\delta \nu + t e_1), \frac12 (\delta 
\nu + t e_1))$ with $S$ contains a right circular cylinder of Euclidean height $\sim \delta$ parallel to $\nu$ and 
radius $\sim (\delta t)^{1/2}$ in the Euclidean-orthogonal directions.
If $k=1$ then $S$ is a null line and $Q$ can be covered by a single causal diamond of timelike diameter as small as we please,  hence $\dims(S)=0$.
For $k>1$, fixing $0<\epsilon<1$ and choosing $t=\delta^{-1+2/\epsilon}$ implies that $J_{\delta,t}$ has timelike diameter $\sim \delta^{1/\epsilon}$ and Euclidean diameter $\sim \delta$ as $\delta \to 0$,
while $K:= J_{\delta,t} \cap S$ has Hausdorff measure 
$\cH^k(K) \gtrsim\delta (\delta t)^{\frac{k-1}{2}} = \delta^{1+(k-1)/\epsilon}$.
It is possible to see that $Q$ can be covered by $\lesssim \delta ^{-1-(k-1)/\epsilon}$ translates of $K$,
thus $\V^{k-1+\epsilon}_\delta(Q) \lesssim 1$. Since $\delta>0$ was arbitrary, $\V^{k-1+\epsilon}(Q)<\infty$ for each 
$0<\epsilon<1$. Thus Corollary \ref{C:dimension} yields $\dims(Q) \le k-1$ as desired.

Finally, we show that $\V^{k-1}$ splits on $S=R\times \R \nu$ as claimed. First, note that by Lemma \ref{L:spacelike 
subspace} we have that $\V^{k-1}\rvert_R = c\mathcal{H}^{k-1}$. Second, it suffices to consider sets of the form $R'\times 
N$, where $N\subseteq \R\nu$ and $R'\subseteq R$. Let $N\subseteq\R\nu$ be countable and let 
$R'\subseteq R$. Then, denoting $N=(n_i)_{i\in\N}$, we have that
\begin{align}
\V^{k-1}(R'\times N) &= \sum_i \V^{k-1}(R'\times\{n_i\})\\
&= c\,\sum_i \mathcal{H}^{k-1}(R') = (c\,\mathcal{H}^{k-1}\times\mathcal{H}^0)(R'\times N)\,.
\end{align}
If $N$ is not countable the same argument shows that $V^{k-1}(R'\times N) \ge c\,\mathcal{H}^{k-1}(R')\cdot\infty = 
(c\,\mathcal{H}^{k-1}\times\mathcal{H}^0)(R'\times N)$.
\end{proof}

\begin{figure}[h!]
%

\tdplotsetmaincoords{70}{30}
\begin{tikzpicture}[scale=0.5, tdplot_main_coords]
\coordinate (O) at (0,0,0);
\draw[dashed, ->] (xyz cs:x=-13.5) -- (xyz cs:x=13.5) node[above] {$\partial_i$};
\draw[dashed, ->] (xyz cs:y=-14.5) -- (xyz cs:y=14.5) node[right] {$\partial_0=\partial_t$};
\draw[dashed, ->] (xyz cs:z=-9.5) -- (xyz cs:z=9.5) node[above] {$\partial_j$};


\fill[green,opacity=0.5] (-6,6,6) -- (6,6,6) -- (6,-4,-4) -- (-6,-4,-4) -- cycle;
\node[green] at (5.8,0,-2) {$S$};

\node[blue, fill,circle,inner sep=1.5pt,text=blue,label={[text=blue,left]:$\frac{1}{2}(\delta\nu + t e_1)$}] (Ap) at 
(-4,0,7) {};
\draw[blue] (Ap) -- (-7,5,-5);
\draw[blue] (Ap) -- (-1,10,-5);

\node[blue, fill,circle,inner sep=1.5pt,text=blue,label={[text=blue,below]:$-\frac{1}{2}(\delta\nu + t e_1)$}] (Am) at 
(5,1,-7) {};
\draw[blue] (Am) -- (-7,15,5);
\draw[blue] (Am) -- (-11,-1,5);

\path [fill=red, opacity=0.7, draw=black] plot [smooth cycle] coordinates {(-3.4,0,0) (-2.3,-0.2,4.3)  (.8,-0.1,2.5) (3,0,0) 
(1,3.2,-5.2)};
\draw[|-|] (-5.5,0.5,-2) -- (0,1.5,-6);
\node[red] at (-2.5,0,-4) {$\delta$};
\draw[|-|] (2,1,-4.7) -- (2.5,3,-4.3);
\node[red, rotate=35] at (3.2,0.8,-4.5) {\footnotesize $\sqrt{\delta t}$};




\end{tikzpicture}
\caption{The {\red intersection} (in red) of the causal cones {\blue $J^\pm(\mp(\delta\nu+t e_1))$} (in blue) with 
the null subspace~{\green $S$} (in green) from the proof of Lemma \ref{L:null subspaces}.}\label{fig-nul-sub}
\end{figure}
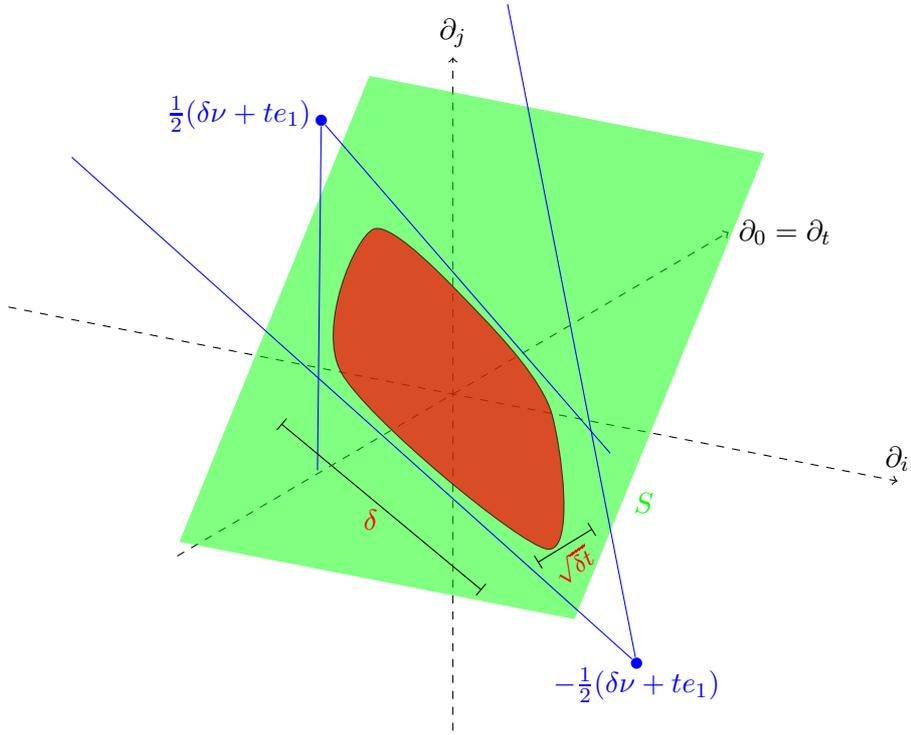

%

\section{Continuous spacetimes}\label{sec-con-st}
Throughout this section $(M,g)$ is a {continuous}, strongly causal and causally plain spacetime of dimension $n$. Here {\em continuous} means that although the manifold $M$ is smooth, the Lorentzian metric tensor $g$ various continuously (but not necessarily smoothly) from point to point.

\subsection{Doubling of causal diamonds}\label{subsec-dou-cd}
In the metric theory \emph{doubling measures} and \emph{doubling metric spaces} are a convenient notion that generalize 
finite dimensional spaces with Lebesgue or Hausdorff measures. Since these concepts are formulated in terms of balls we have 
to adapt it to causal diamonds for our purpose.
\medskip

A main technique in our work is to compare volumes of causal diamonds and suitable \emph{enlargements} of diamonds. Federer, 
in \cite[Subsec.\ 2.8]{Fed:69}, develops a general concept of enlargement of sets, which he then uses in the Carath\'eodory 
construction of measures. However, we need something stronger and so give an axiomatic description here and establish in the 
following subsection that continuous spacetimes are a class of examples where this construction is possible.


\begin{defi}[Enlargement of causal diamonds]
Adopting the setting and notation of Definition \ref{defi-can-mea}, let $\mathcal{F}\subseteq \J$ be a family of 
causal diamonds. Let 
$\varDelta\colon\mathcal{F}\rightarrow[0,\infty)$ be bounded and let $1<\xi<\infty$. Then the $\varDelta,\xi$-enlargement 
$\widehat{J}$ of a causal diamond $J\in\mathcal{F}$ is defined as
\begin{align}
 \widehat{J}:=\bigcup \{J'\in\mathcal{F}:\, J\cap J'\neq\emptyset,\, \varDelta(J')\leq \xi\,\varDelta(J)\}\,.
\end{align}
Let $\mathcal{F}'\subseteq\mathcal{F}$. We call $(\mathcal{F},\mathcal{F}',\varDelta, \xi)$ a \emph{reasonable enlargement} 
\begin{enumerate}
 \item if every point $p$ lies in one causal diamond $J$ of $\mathcal{F}'$, and
 \item if there exists a constant $\Xi\geq 1$ such that for all $J,J'\in \mathcal{F}'$ one has that $J\cap J\neq\emptyset$ 
and $\varDelta(J')\leq \xi\, \varDelta(J)$, then there is a causal diamond $\tilde J\in\mathcal{F}$ with $J,J'\subseteq 
\tilde J$,  $\widehat{J}\subseteq \tilde J$ and $\varDelta(\tilde J)\leq \Xi\, \varDelta(J)$.
\end{enumerate}
Finally, a \emph{doubling} of causal diamonds is any reasonable $(\mathcal{F},\mathcal{F}',\varDelta,2)$-en\-largement.
\end{defi}

\subsection{Cylindrical neighborhoods}\label{subsec-cyl-nhd}
A useful tool for studying continuous spacetimes are \emph{cylindrical neighborhoods} introduced by Chru\'sciel and Grant in \cite[Def.\ 1.8]{CG:12}. We now establish the existence of a refined version of such neighborhoods adapted to our purpose. Moreover, after we have constructed them we will use them throughout this section without always recalling all the details. See also Figures \ref{fig-cyl-nhd}, \ref{fig-cyl-nhd-zoom} for schematic drawings of this constructions.

For a point $p\in M$ and a neighborhood $W$ of $p$ we define $J^\pm(p,W)$ as the \emph{local causal future and past}, 
i.e., $J^\pm(p,W):=\{q\in W:$ there is a future/past directed causal curve connecting $p$ to $q$ that is contained in $W\}$. 
Moreover, for $p,q\in M$ we set $J(p,q,W):=J^+(p,W)\cap J^-(p,W)$.

\begin{lem}[Cylindrical neighborhoods and doubling]\label{lem-adv-cyl-nhd}
Given $C> 1$, every point $p_0\in M$ has a neighborhood $W\subseteq M$ that has the following properties:
 \begin{enumerate}
  \item The neighborhood $W$ is an open, connected, relatively compact coordinate chart such that $\eta_{C^{-1}} {\prec g \prec} \eta_C$ on $W$.
  \item It is cylindrical, i.e., $W=(0,B)\times Z$ and $p_0$ has coordinates $(\frac{B}{2},0)$, where $Z\subseteq \R^{n-1}$.
  \item The coordinate vector field $\partial_t=\partial_{x^0}$ is uniformly timelike on $W$.
  \item There is a smaller, open neighborhood $W'\subseteq (a,b)\times V\subseteq W$ of $p_0$ (with $V\subseteq Z$) that is causally convex in $W$ (i.e., $J(p,q,W)=J(p,q)$ for all $p,q\in W'$) and such that $b-a< \frac{B}{4\lambda}$,  $\forall p=(t,x),q=(s,x)\in W'$: $\hat p=(t- \lambda (s-t),x),\hat q=(s+ \lambda (s-t),x)\in W$, where $\lambda=3 C^2 + 2\geq 5$.
  \item\label{lem-adv-cyl-nhd-enl}
 Moreover, $\forall p=(t,x),q=(s,x), p'=(t',x'),q'=(s',x')\in W'$ with $p\ll q$, $p'\ll q'$, $s'-t'\leq 2 (s-t)$ and $J(p,q)\cap J(p',q')\neq \emptyset$ we have $J(p',q')\subseteq J(\hat p,\hat q,W)\subseteq W$.
  \item Finally, $W$ can be made arbitrarily small and globally hyperbolic.
 \end{enumerate}
 We will refer to $(W',W)$ as \emph{cylindrical neighborhood} of $p_0$.
\end{lem}
\begin{pr}
 By \cite[Prop.\ 1.10(i)]{CG:12} (where $C$ is fixed to $C=2$, but could be an arbitrary constant greater than one) there is a cylindrical neighborhood $W=(0,B)\times Z$ such that $\eta_{C^{-1}}\prec g \prec \eta_C$ on $W$ and $\partial_t$ is timelike on $W$. We now will work in these coordinates in $\R^n$ and suppress for convenience the chart. By construction, $[\frac{B}{4},\frac{3B}{4}]\times Z$ is a neighborhood of $p_0$, so by strong causality there is an open neighborhood $W'$ of $p_0$ that is causally convex in $W$ and contained in $(a,b)\times V\subseteq [\frac{B}{4},\frac{3B}{4}]\times Z$. In particular, we have $\frac{B}{4}\leq a < b \leq \frac{3B}{4}$. Now set $\lambda:=3 C^2 + 2> 5$, and let $p=(t,x),q=(s,x), p'=(t',x'),q'=(s',x')\in W'$ with $p\ll q$, $p'\ll q'$, $s'-t'\leq 2 (s-t)$ and $z=(r,y)\in J(p,q)\cap J(p',q')\neq \emptyset$. Then $p< q'$ and $p'< q$, which implies that $t< s'$ and $t'<s$. Consequently, we also have that $p<_{\eta_C} q'$ and $p'<_{\eta_C} q$ and so
 $|x-x'|\leq C \min(s'-t,s-t')$. Furthermore, $|x-x'|\leq |x-y| + |y-x'|\leq C (r-t + r-t') \leq C (s-t + s'-t')\leq 3C (s-t)$. 
 
 At this point set $v_-:=p'-\hat p$ and $v^+:=\hat q - q'$, then
 \begin{align}
  \eta_{C^{-1}}(v_-,v_-)&\leq \frac{-1}{C^2} (t'-t+\lambda(s-t))^2 + 9 C^2 (s-t)^2\\
  &\leq C^2 (-9 + 9)(s-t)^2=0\,.
 \end{align}
 Similarly, we get $\eta_{C^{-1}}(v_+,v_+)\leq 0$ and thus $J(p',q')\subseteq J(\hat p,\hat q, W)$.
 
 Finally, note that every point has a neighborhood that is globally hyperbolic with respect to a smooth metric with wider light cones, hence it is also globally hyperbolic for $g$ (cf.\ the proof of \cite[Thm.\ 2.2]{SS:18}), so $W$ can be chosen to be globally hyperbolic as well.
\end{pr}

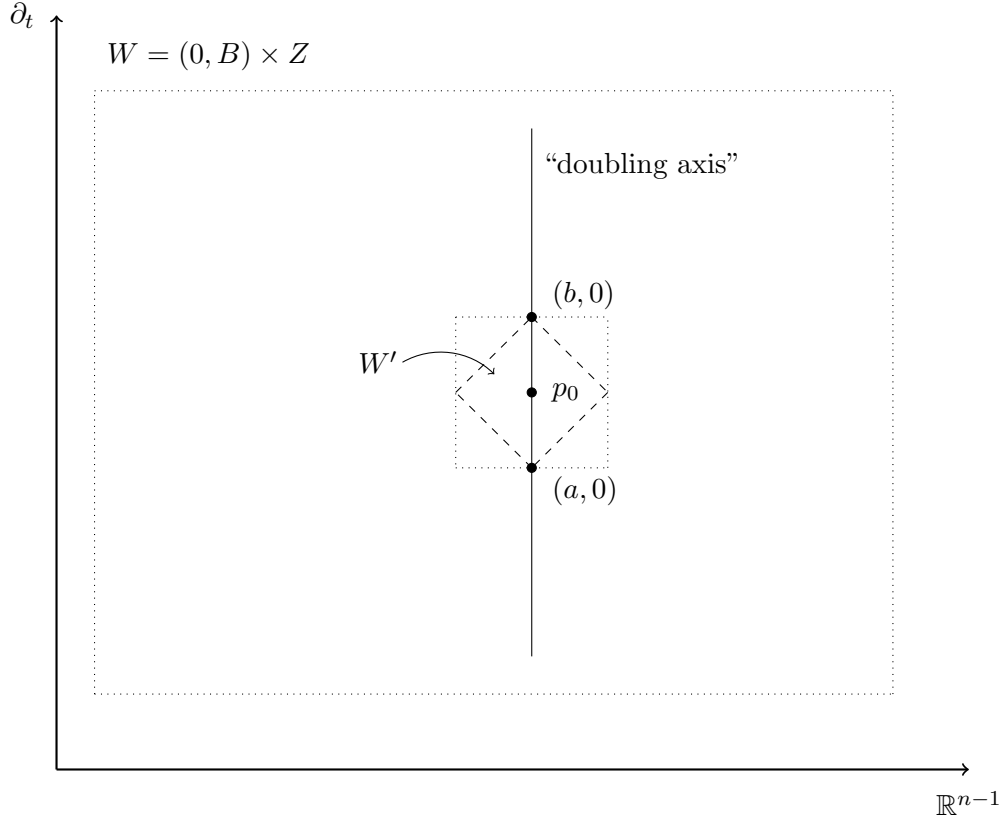
\begin{figure}[h!]
 \begin{tikzpicture}

    \draw[thick, ->] (-3,-5) -- (9, -5) node [label=below:$\R^{n-1}$] {};
    \draw[thick, ->] (-3,-5) -- (-3, 5) node [label=left:$\partial_t$] {};
    
    \draw[dotted, -] (-2.5,-4) -- (8,-4);
    \draw[dotted, -] (-2.5,4) -- (8,4);
    \draw[dotted, -] (-2.5,-4) -- (-2.5,4);
    \draw[dotted, -] (8,-4) -- (8,4);
    \node at (-1,4) [label=above: ${W=(0,B)}\times Z$] {};
    
    \draw [fill] (3.25,0) circle [radius=0.06] node [label=right: $p_0$]{};
    
    \draw [-] (3.25,-3.5) -- (3.25,3.5);
    \node at (4.7,3) {``doubling axis''};
    \draw [fill] (3.25,1) circle [radius=0.06];
    \node at (3.25,1.3) [label=right: ${(b,0)}$]{};
    \draw [fill] (3.25,-1) circle [radius=0.06];
    \node at (3.25,-1.3) [label=right: ${(a,0)}$]{};
    
    \draw[dotted, -] (2.25,-1) -- (4.25,-1);
    \draw[dotted, -] (2.25,1) -- (2.25,-1);
    \draw[dotted, -] (2.25,1) -- (4.25,1); 
    \draw[dotted, -] (4.25,-1) -- (4.25,1);
    
    \draw[dashed, -] (2.25,0) -- (3.25,1);
    \draw[dashed, -] (3.25,1) -- (4.25,0);
    \draw[dashed, -] (4.25,0) -- (3.25,-1);
    \draw[dashed, -] (3.25,-1) -- (2.25,0);
    \node at (0.7,0.4) [label=right: $W'$] {};
    \draw [->] (1.55,0.4) arc [radius=1, start angle=120, end angle= 45];

\end{tikzpicture}
\caption{A schematic drawing of a cylindrical neighborhood $(W,W')$ around~$p_0$ in Lemma \ref{lem-adv-cyl-nhd}
(note $W'$ has to be considerably smaller than $W$).}\label{fig-cyl-nhd}
\end{figure}

\begin{figure}[h!]
 \begin{tikzpicture}
 
%
    
    \node at (-2,4.5) [label=right: $W$] {};
    \draw[dashed, -] (-2.5,0) -- (2.5,5);
    \draw[dashed, -] (2.5,5) -- (7.5,0);
    \draw[dashed, -] (7.5,0) -- (2.5,-5);
    \draw[dashed, -] (2.5,-5) -- (-2.5,0);
    \draw [->] (3.55,4.8) arc [radius=1, start angle=350, end angle= 300];
    \node at (3,5) [label=right: $W'$] {};
    
    \draw[blue, -] (-1.5,0) -- (-0.5,1);
    \draw[blue, -] (-0.5,1) -- (0.5,0);
    \draw[blue, -] (0.5,0) -- (-0.5,-1);
    \draw[blue, -] (-0.5,-1) -- (-1.5,0);
    \draw [fill] (-0.5,1) circle [radius=0.06];
    \node at (-0.5,1) [label=right: $q$]{};
    \draw [fill] (-0.5,-1) circle [radius=0.06];
    \node at (-0.5,-1) [label=right: $p$]{};
    
    \draw[red, -] (-0.5,0) -- (0.5,1);
    \draw[red, -] (0.5,1) -- (1.5,0);
    \draw[red, -] (1.5,0) -- (0.5,-1);
    \draw[red, -] (0.5,-1) -- (-0.5,0);
    \draw [fill] (0.5,1) circle [radius=0.06];
    \node at (0.5,1) [label=right: {$q'=(s',x')$}]{};
    \draw [fill] (0.5,-1) circle [radius=0.06];
    \node at (0.5,-1) [label=right: {$p'=(t',x')$}]{};

    \draw [-] (-0.5,-4) -- (-0.5,4);
    \node at (0.9,4) {``doubling axis''};
    \draw[dotted] (-0.5,-4) -- (-0.5,-4.8);
    \node at (-0.5,-4.8) [below=0.1pt] {$x$};
    
    \draw[blue, -, dashed] (-4,0) -- (-0.5,3.5);
    \draw[blue, -, dashed] (-0.5,3.5) -- (3,0);
    \draw[blue, -, dashed] (3,0) -- (-0.5,-3.5);
    \draw[blue, -, dashed] (-0.5,-3.5) -- (-4,0);
    \draw [fill] (-0.5,3.5) circle [radius=0.06];
    \node at (-0.5,3.5) [label=right: $\hat q$]{};
    \draw [fill] (-0.5,-3.5) circle [radius=0.06];
    \node at (-0.5,-3.5) [label=right: $\hat p$]{};
    
    \draw[|-|] (-4,3.5) -- (-4,1);
    \node at (-4,2.25) [text=blue, left=1pt]{$\lambda\varDelta$};
    \draw[|-|] (-4,1) -- (-4,-1);
    \node at (-4,0) [left=3pt]{$\varDelta=s-t$};
    \draw[|-|] (-4,-1) -- (-4,-3.5);
    \node at (-4,-2.25) [text=blue, left=1pt]{$\lambda\varDelta$};
    
\end{tikzpicture}
\caption{Doubling inside a cylindrical neighborhood from Corollary \ref{cor-con-enl}.}
\label{fig-cyl-nhd-zoom}
\end{figure}
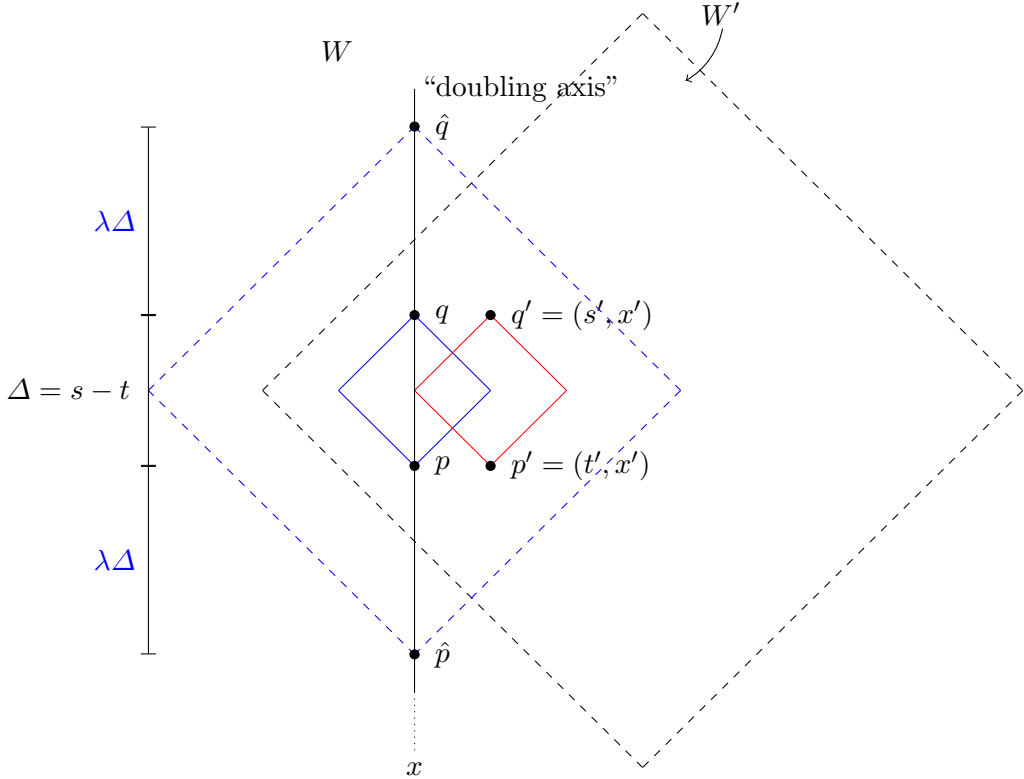

We will now work in such refined cylindrical neighborhoods and refer to $W', W$, $C,\lambda$, $\hat p,\hat q$ etc.\ without 
recalling their definition every time.

\begin{cor}[Cylindrical neighborhoods yield a doubling of causal diamonds]\label{cor-con-enl}
 Let $(W,W')$ be a cylindrical neighborhood as above. Setting $\mathcal{F}^{(')}:=\{J(p,q): p=(t,x),q=(s,x)\in W^{(')}$ with 
$p\ll q \}$ and $\varDelta(J((t,x),(s,x))):=s-t$, yields a reasonable $(\mathcal{F},\mathcal{F}',\varDelta,2)$-enlargement, 
hence a doubling of causal diamonds. Moreover,
\begin{equation}
  \widehat{J(p,q)}\subseteq W\cap J(\hat p,\hat q, W)\subseteq W\cap J(\hat p,\hat q)\,,
 \end{equation}
 for all $p=(t,x),q=(s,x)\in W'$.
\end{cor}
\begin{pr}
 Lemma \ref{lem-adv-cyl-nhd}\ref{lem-adv-cyl-nhd-enl} gives that $J(p',q')\subseteq \widehat{J(p,q)}\subseteq W\cap J(\hat 
p,\hat q, W)\subseteq W\cap J(\hat p,\hat q)$ and $\varDelta(J(\hat p,\hat q)) = (1+2\lambda)\, (s-t) = (1+2\lambda)\, 
\varDelta(J(p,q))$ for all $p=(t,x),q=(s,x),p'=(t',x'), q'=(s',x')\in W'$.
\end{pr}

\subsection{Compatibility for spacetimes}\label{subsec-com-st}
The aim of this subsection is to establish that for a continuous, strongly causal and causally plain ${n}$-dimensional spacetime (viewed as a \LLSn) the measure $\V^n$ constructed above in Proposition \ref{pro-con-mea} agrees with the Lorentzian volume measure $\vol$ given by $\vol(A):=\int_A \sqrt{|\det(g)|}\, d x^0\wedge\ldots\wedge x^{n-1}$. In what follows we use the fact that for spacetimes with continuous metrics that are causally well-behaved (i.e., strongly causal and causally plain) the different notions of causal curves and causality conditions agree, cf.\ \cite{KS:18}.
\medskip

%

As a first application of the doubling of causal diamonds introduced in Subsection \ref{subsec-dou-cd} and cylindrical 
neighborhoods, we establish that the volume measure of a continuous spacetime is doubling in a suitable sense.

\begin{lem}[Doubling property of the volume measure]\label{lem-adv-cyl-nhd-dou}
Let $(W,W')$ be a cylindrical neighborhood that is so small that in these coordinates $|\det(g)|$ is bounded above and 
below on $W$. Then there is a constant $L\geq 1$ (depending on on the dimension $n$, the constant $C$ and on the minimum and 
maximum of $|\det(g)|$ on $\overline W$) such that
 \begin{equation}
  \vol^g(\widehat{J(p,q)})\leq L \,\vol^g(J(p,q))\qquad  \forall p=(t,x),q=(s,x) \in W'\,.
 \end{equation}
\end{lem}


\begin{pr}
 In these coordinates the volume measure $\vol^g$ is just a $|\det(g)|d\mathcal{H}^n$, where $\mathcal{H}^n$ is the $n$-dimensional Hausdorff measure on $\R^n$. By assumption, $0<k\leq |\det(g)|\leq K$ on $W$. Let $A>0$ then $\mathcal{H}^n(J_{\eta_A}(p,q))= A^{n-1} \omega_n \tau_{\eta}(p,q)^n$. Thus we obtain for $p=(t,x),q=(s,x)\in W'$ and $\hat p=(t- \lambda (s-t),x),\hat q=(s+ \lambda (s-t),x)$ that
 \begin{align}
  \omega_n C^{1-n} (s-t)^n = \mathcal{H}^n(J_{\eta_{C^{-1}}}(p,q))\leq \mathcal{H}^n(J(p,q))\,, 
 \end{align}
 as $\eta_{C^{-1}}\prec g$. This yields that
 \begin{align}
  \mathcal{H}^n(J(\hat p,\hat q))&\leq \mathcal{H}^n(J_{\eta_C}(\hat p,\hat q)) 
  \\&= C^{n-1} \omega_n(2\lambda + 1)^n\,(s-t)^n\\
  &\leq (2\lambda+1)^n C^{2(n-1)}\mathcal{H}^n(J(p,q))\,,
 \end{align}
and so by Corollary \ref{cor-con-enl} we conclude that
 \begin{align}
 \vol^g(\widehat{J(p,q)}) &\leq \vol^g(J(\hat p,\hat q, W)) 
 \\&\leq K \mathcal{H}^n(J(\hat p,\hat q, W))\\
 &\leq K (2\lambda+1)^n C^{2(n-1)}\mathcal{H}^n(J(p,q))\\
 &\leq \frac{K}{k}(2\lambda+1)^n C^{2(n-1)}\vol^g(J(p,q))\,.
 \end{align}
\end{pr}

Following \cite[Subsec.\ 2.8, 2.10]{Fed:69} we set up the machinery needed to compare the volume measure with $\V^n$, which is constructed from the causal diamonds by Carath\'eodory's construction.
\begin{defi}[Fine cover]
 Let $(X,d)$ be a metric space, let $\F$ be a family of closed subsets of $X$ and let $A\subseteq X$. We say that \emph{$\F$ covers $A$ finely} if for all $a\in A$ and for all $\eps>0$ there is $F\in\F$ with $a\in F\subseteq B_\eps(a)$.
\end{defi}

\begin{lem}[Controlling the time separation locally]\label{lem-con-tau-loc}
 For every $\delta>0$, $p_0\in M$ and $C>1$ there is a cylindrical neighborhood $W=W(p_0,\delta,C)$ of $p_0$ such that $\eta_{C^{-1}}\prec g \prec \eta_C$ on $W$ and $\diam(W)\leq \delta$. For $p=(t,x),q=(s,x)\in W'$ with $p\leq q$ we have
 \begin{equation}\label{eq-con-tau-loc}
 \sqrt{1-\delta}\, (s-t) \leq \tau(p,q) \leq (\sqrt{(1+C^2)\delta+2C^2-1}\,(s-t)\,.
 \end{equation}
\end{lem}
\begin{pr}
By Lemma \ref{lem-adv-cyl-nhd} we know that for every $\delta>0$, $p_0\in M$ and $C>1$ there is a cylindrical neighborhood $W=W(p_0,\delta,C)$ of $p_0$ such that $\eta_{C^{-1}}\prec g \prec \eta_C$ on $W$, $\diam(W)\leq \delta$ and $|g-\eta|<\delta$. Now let $p=(t,x),q=(s,x)\in W'$ with $p< q$ (in case $p=q$ there is nothing to show) and consider the curve $\lambda\colon[t,s]\rightarrow W$, $\lambda(r):=(r,x)$. Then $\lambda$ is $\eta$- and $g$-future directed timelike from $p$ to $q$. Thus
 \begin{align}
  \tau(p,q)\geq L^g(\lambda) = \int_t^s\sqrt{-(g_{00}+1)+1} > \sqrt{1-\delta}\,(s-t)\,,
 \end{align}
where we used that $|g-\eta|<\delta$ on $W$ and $g_{00}=g(\partial_t,\partial_t)$.
  
Denoting by $|\cdot|_e$ the euclidean norm in this chart we have for any $g$-causal curve $\gamma(r)=(r,\vec\gamma(r))$ in $W$ that $|\dot{\vec{\gamma}}|_e^2\leq C^2$ as such a curve is $\eta_C$-causal. Thus $|\dot\gamma|_e$ is uniformly bounded by $\sqrt{1+C^2}$ for all such curves. Let $\gamma\colon[t,s]\rightarrow W$ be a future directed maximal causal curve in $W$ from $p$ to $q$ (exists by \cite[Thm.\ 2.2]{SS:18}), parametrized as $\gamma(r)=(r,\vec\gamma(r))$. Then
\begin{align}
 \tau(p,q)=L^g(\gamma) &\leq \int_t^s\sqrt{\delta (1+C^2)+(C^2-1) - \eta_C(\dot\gamma,\dot\gamma})\\
 &\leq \sqrt{(1+C^2)\delta +2C^2-1}\,(s-t). 
\end{align}
\end{pr}

For smooth spacetimes we have the following expansion using normal coordinates of the volumes of small causal diamonds (cf.\ e.g.\ \cite[Eq.\ (74)]{GS:07})
\begin{equation}\label{eq-vol-exp}
 \vol^g(J(p,q))=\omega_n \tau(p,q)^n (1 + O(\tau(p,q)^2))\,,
\end{equation}
from which   
\begin{equation}
   \lim_{\tau(p,q)\to 0}\frac{\vol^g(J(p,q))}{\rho_n(J(p,q))}=1\,,
\end{equation}
follows. The latter holds also for continuous spacetimes as the following Lemma establishes.
\medskip

\begin{lem}[Local metric vs geometric volume]
\label{lem-vol-den}
 For every $p_0\in M$
  \begin{equation}
   \lim_{\substack{\diam(J(p,q))\to 0\\p_0\in J(p,q)\in\J'}}\,\frac{\vol^g(J(p,q))}{\rho_n(J(p,q))}=1\,.
  \end{equation}
\end{lem}
\begin{pr}
By Lemma \ref{lem-adv-cyl-nhd} we know that for every $\delta>0$, $p_0\in M$ and $C>1$ there is a cylindrical neighborhood $W=W(p_0,\delta,C)$ of $p_0$ such that $\eta_{C^{-1}}\prec g \prec \eta_C$ on $W$, $\diam(W)\leq \delta$ and $|g-\eta|<\delta$. Now let $p=(t,x),q=(s,x)\in W'$ with $p\leq p_0\leq q$. As $W'$ is causally convex in $W$ we have that $\diam(J(p,q))\leq \diam(W)\leq \delta$ and $J_{\eta_{C^{-1}}}(p,q)\subseteq J(p,q)\subseteq J_{\eta_C}(p,q)$. Consequently, we get that
\begin{align}
C^{1-n}\omega_n (s-t)^n
&= \mathcal{H}^n(J_{\eta_{C^{-1}}}(p,q))
\\&\leq \mathcal{H}^n(J(p,q))
\\&\leq \mathcal{H}^n(J_{\eta_C}(p,q))\\
&= C^{n-1} \omega_n (s-t)^n\,.
\end{align}
So letting the neighborhoods shrink and simultaneously $C\searrow 1$ yields that
  \begin{equation}
   \lim_{\substack{\diam(J(p,q))\to 0\\p_0\in J(p,q)\in\J'}}\,\frac{\mathcal{H}^n(J(p,q))}{\omega_n (s-t)^n}=1\,,
  \end{equation}
where $p=(t,x)$, $q=(s,x)\in W'$. From Equation \eqref{eq-con-tau-loc} we get that $\sqrt{1-\delta}\,(s-t)\leq \tau(p,q)$ and $|det(g)| = 1 + o(1)$ as $\delta\to 0$, thus $\vol^g(J(p,q))=\mathcal{H}^n(J(p,q))(1+o(1))$ and so 
\begin{equation}
\limsup_{\substack{\diam(J(p,q))\to 0\\p_0\in J(p,q)\in\J'}}\,\frac{\vol^g(J(p,q))}{\rho_n(J(p,q))}\leq 1\,.
  \end{equation}
Again by Equation \eqref{eq-con-tau-loc} we bound $\tau$ from above and hence similarly to the estimate for the limit superior we obtain
\begin{equation}
\frac{\vol^g(J(p,q))}{\rho_n(J(p,q))}\geq \frac{\mathcal{H}^n(J(p,q))(1 + o(1))}{\omega_n (s-t)^n (\sqrt{(1+C^2)\delta 
+2C^2-1}\bigr)^n}\to 1\,
  \end{equation}
  as $\delta \to 0$ and $C\to 1$.
Thus
\begin{equation}
\liminf_{\substack{\diam(J(p,q))\to 0\\p_0\in J(p,q)\in\J'}}\,\frac{\vol^g(J(p,q))}{\rho_n(J(p,q))}\geq 1\,,
  \end{equation}
  which finishes the proof.
\end{pr}

\begin{thm}[Metric versus geometric volume]
\label{thm-mea-vol-equ}
 The volume measure $\vol^g$ agrees with the measure $\V^n$ constructed in Proposition \ref{pro-con-mea}, i.e., 
$\vol^g=\V^n$, where the $n$-dimensional spacetime $(M,g)$ is viewed as a \LLSn.
\end{thm}
\begin{pr}
 First, we show equality locally, i.e., in a cylindrical neighborhood $(W',W)$ as in Lemma \ref{lem-adv-cyl-nhd-dou}. To this 
end let $O\subseteq W'$ be open and $A\subseteq O$ Borel measurable. 
 Note that $\J'$ covers $A$ finely, all $J(p,q)$ are closed and $\rho_n({J(\hat p,\hat q)})\leq \tilde L\, \rho_n(J(p,q))$ for all $p=(t,x), q=(s,x)\in W'$ by Equation \eqref{eq-con-tau-loc}. Letting $\eps>0$, for every $a\in A$ Lemma \ref{lem-vol-den} yields
 \begin{equation}\label{eq-mea-vol-equ-lb}
  \limsup_{\substack{\diam(J(p,q))\to 0\\a\in J(p,q)\in\J'}} \frac{\vol^g(J(p,q))}{\rho_n(J(p,q))}>1-\eps\,.
 \end{equation}
 At this point we can apply a slight modification of \cite[Thm.\ 2.10.18(1)]{Fed:69} to get that $\vol^g(O)\geq (1-\eps) \V^n(A)$ (the modification is needed as $\widehat{J(p,q)}\notin\J$ in general). As Equation \eqref{eq-mea-vol-equ-lb} holds for all $\eps>0$ and $\vol^g$ is Borel regular we get that $\vol^g(A)\geq \V^n(A)$ for all $A\subseteq U$ Borel measurable (any Borel measure on a manifold is regular, so also $\V^n$ is; cf.\ Ulam's Theorem \cite[VIII.\S 1 Thm.\ 1.16]{Els:18}).

For the other inequality we again apply Lemma \ref{lem-vol-den} to obtain that for $\eps>0$ we have that for every $a\in A$
 \begin{equation}
  \limsup_{\substack{\diam(J(p,q))\to 0\\a\in J(p,q)\in\J'}} \frac{\vol^g(A\cap J(p,q))}{\rho_n(J(p,q))}<1+\eps\,.
 \end{equation}
Then \cite[Thm.\ 2.10.17(2)]{Fed:69} gives $\vol^g(A)\leq \V^n(A)$ and we are done. Finally, as now $\V^n$ agrees with $\vol^g$ on all Borel measurable subsets of such cylindrical neighborhoods we have that $\V^n=\vol^g$: Let $A\subseteq M$ be Borel measurable and let $(W_l)_l$ be a covering of $M$ by such cylindrical neighborhoods as above. Set $V_0:=W_0$ and $V_l:=W_l\backslash \bigcup_{k=0}^{l-1}V_k$ for $l\geq 1$. Then $V_l\subseteq W_l$ for all $l\in\N$, $(V_l)_l$ is pairwise disjoint and $M=\bigcup_l V_l$. Thus
\begin{equation}
 \vol^g(A)=\sum_l \vol^g(A\cap V_l) = \sum_l \V^n(A\cap V_l) = \V^n(A)\,.
\end{equation}
\end{pr}

\begin{prop}[Topological vs geometric dimension for spacetimes]
\label{prop-con-st-d-uni-dim}
 Let $(M,g)$ be a strongly causal, causally plain continuous spacetime of dimension $n$ and $h$ a smooth Riemannian background metric on $M$. Then
 \begin{enumerate}
  \item the induced \LpLS is locally $d^h$-uniform.
  \item Moreover, the geometric dimension $\dims(M)$ agrees with the manifold dimension, i.e. $\dims(M)=n$.
 \end{enumerate}
\end{prop}
\begin{pr}
  (i) This follows from the proof of Lemma \ref{lem-con-tau-loc}, in particular from Equation \eqref{eq-con-tau-loc}, and the fact that locally in a chart, one can always bound the  Euclidean distance by the induced Riemannian distance.
  
  (ii) Theorem \ref{thm-mea-vol-equ} gives that $\V^n=\vol^g$. Let $M=\bigcup_i U_i$ be a covering of $M$ by open charts such 
that in each chart $U_i$ the volume measure $\vol^g$ is absolutely continuous with respect to the $n$-dimensional Lebesgue 
measure (cf.\ Lemma \ref{lem-adv-cyl-nhd-dou}) and $\vol^g(U_i)<\infty$. Then, $\dims(U_i)\leq n$ for all $i\in\N$. Also, 
assuming $\dims(U_i)<n$, Lemma \ref{lem-dim-mM},\ref{lem-dim-mM-M} gives that $0=\V^n(U_i)=\vol^g(U_i)$, which contradicts 
the absolute continuity with respect to the $n$-dimensional Lebesgue measure. Finally, this yields $\dims(U_i)=n$ for all 
$i\in\N$ and so by Lemma \ref{lem-dim-cup}, $\dims(M)=n$.
\end{pr}

\subsection{Doubling measures on continuous spacetimes}\label{subsec-cou-dou-mea}
A (Borel) measure $\mu$ on a metric space $(X,d)$ is said to be \emph{doubling} if there exists a constant $C\geq 1$ (called the \emph{doubling constant}) such that for all $x\in X, r>0$ one has that
\begin{equation}\label{eq-met-dou}
 \mu(B_{2r}(x))\leq C\, \mu(B_r(x))\,.
\end{equation}
By restricting to small balls one arrives at an analogous notion of \emph{local doubling measure}. Of particular interest for 
us is the following classical result: If $(X,d)$ is a metric space with a doubling measure that has doubling constant $C$, 
then the Hausdorff dimension of $(X,d)$ is bounded above by $\log_2(C)$, i.e., $\dim^H(X)\leq \log_2(C)$, cf.\ e.g.\ 
\cite[Cor.\ 2.5]{Stu:06b}. Moreover, a synthetic (lower) bound on the Ricci curvature implies that the reference measure is 
locally doubling, hence gives a bound on the Hausdorff dimension, cf.\ e.g.\ \cite[Cor.\ 2.4]{Stu:06b}. Thus, a natural 
question in the Lorentzian setting is what is the right analog of the doubling property \eqref{eq-met-dou} and how does it 
relate to the geometric dimension introduced in Section \ref{sec-dim}? We will define local doubling for causal diamonds in 
the setting of continuous spacetimes below and relate the (causal) doubling constant to the geometric dimension in Theorem 
\ref{thm-dou-syn-dim}.
\bigskip

Lemma \ref{lem-adv-cyl-nhd-dou} indicates what the (local) doubling property of the volume measure $\vol^g$ is and thus we are lead to the following:

\begin{defi}[Causal doubling property]\label{def-cau-dou}
Let $\Xll$ be a \LpLSn, let $\mathcal{F}\subseteq\J$ be a family of causal diamonds and let 
$(\mathcal{F},\mathcal{F}',\varDelta,2)$ be a doubling of causal diamonds. A Borel measure $m$ on $(X,d)$ is 
\emph{causally doubling} if there exists a constant $L\geq 1$ such that
\begin{enumerate}
 \item for all $x,y\in X$ one has
 \begin{equation}\label{eq-cau-dou}
  m(\widehat{J(x,y)})\leq L\, m(J(x,y))\,, \text{and}
 \end{equation}
 \item $0<m(J(x,y))<\infty$ for $x,y\in X$ with $x\ll y$.
\end{enumerate}
The constant $L$ is called the \emph{doubling constant of $m$}.
\end{defi}

One can easily also introduce a local version of the causal doubling property \ref{def-cau-dou}. However for our purposes we 
will focus on the case of continuous spacetimes and therefore give a definition adapted to this setting. In essence it 
says that if we restrict the measure to the doubling of causal diamonds in a cylindrical neighborhood, cf.\ Corollary 
\ref{cor-con-enl}, then this restriction is a causally doubling.

\begin{defi}[Local doubling property]
A Borel measure $m$ on a continuous spacetime $(M,g)$ is \emph{locally causally doubling} if for every cylindrical 
neighborhood $(W',W)$ there exists a constant $L\geq 1$ (depending on $n,\overline{W},C)$ such that
\begin{enumerate}
 \item for all $p=(t,x),q=(s,x)\in W'$ one has
 \begin{equation}\label{eq-loc-cau-dou}
  m(J(\hat p,\hat q, W))\leq L\, m(J(p,q))\,,
 \end{equation}
 \item $0<m(J(p,q,W))<\infty$ for $p,q\in W$ with $p\ll q$,
 \item $m(\overline{W})<\infty$.
\end{enumerate}
The constant $L$ is called the \emph{(local) doubling constant of $m$} (with respect to $(W',W)$).
\end{defi}


\begin{thm}[Ratio of measures via local doubling]\label{thm-dou-rat}
 Let $m$ be locally causally doubling. Then for any cylindrical neighborhood $(W',W)$ and for $\tilde W=(\tilde a, \tilde b)\times \tilde V\subseteq W'$ open, non-empty, there are constants $K,\kappa>0$ (depending only on $C,\delta$ and $L$) such that the following holds: if $p_0=(t_0,x_0),q_0=(s_0,x_0),p=(t,x),q=(s,x)\in \tilde W$ with $J(p,q)\cap J(p_0,q_0)\neq\emptyset$ and $0<s_0-t_0 < 2 \frac{\tilde b-\tilde a}{1+2\lambda}$ with
 \begin{equation}\label{eq-thm-dou-rat-ass}
  s+t\in \bigl(2\tilde a+\frac{s_0-t_0}{2}(1+2\lambda),2\tilde b-\frac{s_0-t_0}{2}(1+2\lambda)\bigr)\,,
 \end{equation}
then
 \begin{equation}
  \frac{m(J(p,q))}{m(J(p_0,q_0))} \geq \frac{1}{K} \Bigl(\frac{\tau(p,q)}{\tau(p_0,q_0)}\Bigr)^{\kappa}\,.
 \end{equation}
\end{thm}
\begin{pr}
For $k\in\N$ we denote by $(\hat p^k, \hat q^k)$ the $k$-times ``enlarged'' points, i.e., $\hat p^k=(\hat t^k, x)$, $\hat q^k=(\hat s^k,x)$ where $\hat t^k = t+\frac{s-t}{2}(1-(1+2\lambda)^k)=\frac{s+t}{2}-\frac{s-t}{2}(1+2\lambda)^k$ and $\hat s^k = s-\frac{s-t}{2}(1-(1+2\lambda)^k)=\frac{s+t}{2}+\frac{s-t}{2}(1+2\lambda)^k$. We want to find $k\in\N$ minimal such that $J(p_0,q_0)\subseteq J(\hat p^k,\hat q^k, W)$. By Lemma \ref{lem-adv-cyl-nhd}\ref{lem-adv-cyl-nhd-enl} it suffices to find a $k$ such that
\begin{equation}\label{eq-k-enl-pts}
 2 (\hat s^{k-1}-\hat t^{k-1})< s_0-t_0 \leq 2 (\hat s^k - \hat t^k)\,,
\end{equation}
and this is possible as by assumption \eqref{eq-thm-dou-rat-ass} since $J(p,q)\cap J(p_0,q_0)\neq\emptyset$. Thus let 
$k\in\N$ satisfy Equation \eqref{eq-k-enl-pts}. By construction we have that $\hat p^{k},\hat q^{k}\in \tilde W$, so
\begin{equation}
 J(p_0,q_0)\subseteq \reallywidehat{J(\hat p^{k-1},\hat q^{k-1})}\subseteq J(\hat p^k,\hat q^k, W)\subseteq W\,.
\end{equation}
The doubling property (on $W$, with local doubling constant $L$) implies that
\begin{equation}\label{eq-k-dou-est}
 m(J(p_0,q_0))\leq m(J(\hat p^k,\hat q^k)) \leq L^k m(J(p,q))\,.
\end{equation}
Since $2(1+2\lambda)^{k-1}(s-t) = 2 (\hat s^{k-1} - \hat t^{k-1}) < s_0-t_0$ and setting $\kappa:=\log_{1+2\lambda}(L)$ we conclude that
\begin{equation}
 \Bigl(\frac{s-t}{s_0-t_0}\Bigr)^\kappa < \underbrace{\Bigl(\frac{1+2\lambda}{2}\Bigr)^\kappa}_{=:K}(1+2\lambda)^{-k\kappa} = K L^{-k}\,.
\end{equation}
This gives by Equation \eqref{eq-k-dou-est} that
\begin{equation}
  \frac{m(J(p,q))}{m(J(p_0,q_0))} \geq \frac{1}{K} \Bigl(\frac{s-t}{s_0-t_0}\Bigr)^{\kappa}\,.
\end{equation}
Finally, we apply Lemma \ref{lem-con-tau-loc} to estimate $\frac{s-t}{s_0-t_0}$ by $\frac{\tau(p,q)}{\tau(p_0,q_0)}$ and thereby only changing the constant $K$.
\end{pr}

Following \cite[Subsec.\ 2.8]{Fed:69} we set up the machinery needed to get a \emph{Vitali covering lemma} for causal diamonds and causally doubling measures. First we recall the following

\begin{defi}[Adequate]
 Let $(X,d)$ be a metric space and let $\nu$ be a Borel measure on $(X,d)$ such that any bounded set has finite 
$\nu$-measure. Let $\F$ be a family of closed subsets of $X$ and let $A\subseteq X$. We say that $\F$ is 
\emph{$\nu$-adequate for $A$} if for all open $V\subseteq X$ there is a countable subfamily $\G\subseteq\F$ of pairwise 
disjoint sets with $\bigcup_{G\in\G} G\subseteq V$ and
 \begin{equation}
  \nu((V\cap A)\backslash \bigcup_{G\in\G}G)=0\,.
 \end{equation}
\end{defi}

\begin{prop}[Doubling criterion for $m$-adequacy]
\label{prop-VCL-CD}
Let $m$ be a locally causally doubling measure on (a continuous, strongly causal, causally plain spacetime) $M$. Given a cylindrical neighborhood $(W',W)$ in $M$, if $A\subseteq W'$, then $\J'$ is $m$-adequate for $A$.
\end{prop}
\begin{pr}
Every causal diamond in $\J'$ is closed and contained in $W$. Moreover, any $A\subseteq W'$ is finely covered by $\J'$: if $a=(r,z)\in A$ and $\eps>0$ then by strong causality there are $\bar p, \bar q$ such that $a\in I(\bar p,\bar q)\subseteq J(\bar p,\bar q)\subseteq B_\eps(a)$. Then let $\alpha>0$ be small enough that $a_-:=(r-\alpha,z)\ll a\ll a_+:=(r+\alpha,z)\in I(\bar p,\bar q)$ and $a_\pm\in W'$. Then $J(a_-,a_+)\in\J'$ and $a\in J(a_-,a_+)\subseteq J(\bar p,\bar q)\subseteq B_\eps(a)$.

The doubling property and Lemma \ref{lem-adv-cyl-nhd-dou} yield that
\begin{equation}
 m(\widehat{J(p,q)})\leq m(J(\hat p,\hat q, W))\leq L\, m(J(p,q))\,,
\end{equation}
where $L$ is the local doubling constant of $m$. Thus \cite[Thm.\ 2.8.7]{Fed:69} applies and yields that $\J'$ is $m$-adequate for $A$. 
\end{pr}

From \cite[Thm.\ 2.8.4, Cor. 2.8.5]{Fed:69} we deduce in the above setting the following useful Corollary, which could be of independent interest.

\begin{cor}[Covering by enlargements of disjoint diamonds]
\label{cor-dis-enl}
In the above setting let $A\subseteq W'$ and $G\subseteq\J'$ be $m$-adequate for $A$. Then there is a family of pairwise 
disjoint causal diamonds $\J''\subseteq G$ such that
 \begin{equation}
  \bigcup_{J\in J''} J\ \subseteq \bigcup_{J\in G}\widehat{J}\,,
 \end{equation}
 and for any finite subset $H\subseteq J''$ one has that
 \begin{equation}
  A\ \backslash\bigcup_{J\in H} J \subseteq \bigcup_{J\in \J''\backslash H} \widehat{J}\,.
 \end{equation}
\end{cor}

Analogously to the metric case (see e.g.\ \cite[Cor.\ 2.5]{Stu:06b}) a doubling property implies a bound on the (synthetic) dimension.
\begin{thm}[Doubling constant bounds the geometric dimension]\label{thm-dou-syn-dim}
 Let $(M,g)$ be a continuous, causally plain and strongly causal spacetime of dimension $n$. Let $m$ be a locally causally doubling measure on the induced \LLS of $(M,g)$, which has the same local doubling constant $L$ on all sufficiently small cylindrical neighborhoods.
 Then
 \begin{equation}
  n=\dims(M)\leq \log_{1+2\lambda}(L)\,.
 \end{equation}
\end{thm}
\begin{pr}
 Set $\kappa:=\log_{1+2\lambda}(L)$.  By Lemma \ref{lem-dim-cup} it suffices to show that $\V^\kappa(J)$ is uniformly bounded (independent of $J$) for all causal diamonds in $M$, as $M$ can be written as a countable union of causal diamonds. Moreover, it suffices to only consider sufficiently small causal diamonds $J=J(p_0,q_0)$ in a cylindrical neighborhood $(W',W)$ (with the same fixed constants $C>1$ and $\lambda=3C^2+2$). To be precise, let $\tilde W=(\tilde a, \tilde b)\times \tilde V\subseteq W'$ open, non-empty and let $p_0=(t_0,x_0),q_0=(s_0,x_0)\in\tilde W$ with $0<s_0-t_0 < 2 \frac{\tilde b-\tilde a}{1+2\lambda}$ and $J\subseteq \tilde W$.
 
 Note that for any $\bar p=(\bar t,\bar x),\bar q=(\bar s,\bar x)\in W$ we have $\diam(J(\bar p,\bar q, W))\leq 2\sqrt{1+C^2}\,(\bar s - \bar t)$ by a similar calculation to the one in the proof of Lemma \ref{lem-con-tau-loc}.

  For $0<\xi<{4\sqrt{1+C}\,(1+2\lambda)}\min (t_0-\tilde a,\tilde b - s_0)$, set $T_\xi:=\frac{\xi}{2 \sqrt{1+C}\,(1+2\lambda)}$. Let $J_i:=J(p_i, q_i)$ ($i\in I_\xi$) be \emph{maximally $T_\xi$-separated}, i.e., the family satisfies
 \begin{enumerate}
  \item $p_i=(t_i,x_i)$, $q_i=(s_i,x_i)\in \tilde W$,
  \item $s_i-t_i = T_\xi$,
  \item for all $i,j\in I_\xi$, $i\neq j$ one has $p_i\not\leq q_j$ or $|s_j-t_i|>2 T_\xi$ or $p_j\not\leq q_i$ or $|s_i-t_j|>2 T_\xi$, and finally
  \item $J_i\cap J\neq\emptyset$ for all $i\in I_\xi$.
 \end{enumerate}
This implies that $\diam(J_i)\leq\xi$ and $\diam(J(\hat p_i,\hat q_i))\leq \xi$ for all $i\in I_\xi$. By Lemma \ref{lem-con-tau-loc} we deduce that $\tau(p_i,q_i)\geq C_1 \xi$ and $\tau(\hat p_i,\hat q_i)\leq C_2\xi$. Furthermore, the sets $(J_i)_{i\in I_\xi}$ are disjoint and one has $J\subseteq \bigcup_{i\in I_\xi}J(\hat p_i, \hat q_i)$. The latter follows from a similar calculation to the proof of Lemma \ref{lem-adv-cyl-nhd}.

 Finally, we choose $t_i\in \bigl(\tilde a-\frac{T_\xi}{2}+\frac{s_0-t_0}{4}(1+2\lambda),\tilde b - \frac{T_\xi}{2} -\frac{s_0-t_0}{4}(1+2\lambda)\bigr)$ ($i\in I_\xi$). Then by Theorem \ref{thm-dou-rat} there are constants $\kappa>0, K>0$ (only depending on $C,\delta, L$) such that for all $i\in I_\xi$ we have
\begin{equation}
 m(J_i)\geq \tilde K\, \tau(p_i,q_i)^\kappa\,,
\end{equation}
as $m(J)<\infty$. 

At this point we estimate, using the pairwise disjointness of the causal diamonds $(J_i)_{i\in I}$,
\begin{align}
 \infty > m(\overline{W}) \geq m(\bigcup_{i\in I_\xi} J_i) = \sum_{i\in I_\xi} m(J_i) \geq \tilde K \sum_{i\in I_\xi} \tau(p_i,q_i)^\kappa \geq \tilde K C_1^\kappa \xi^\kappa |I_\xi|\,.
\end{align}
This gives $|I_\xi|\leq C_3 \xi^{-\kappa}$, where we subsumed all the constants into $C_3$. Finally, we are in the position to estimate
\begin{align}
 \V^\kappa_\xi(J)\leq\sum_{i\in I_\xi} \rho_\kappa(J(\hat p_i,\hat q_i)) = \omega_\kappa \sum_{i\in I_\xi} \tau(\hat p_i,\hat q_i)^\kappa \leq \omega_\kappa |I_\xi| C_2^\kappa \xi^\kappa \leq \omega_\kappa C_3 C_2^\kappa<\infty\,,
\end{align}
where the constants on the left-hand-side do not depend on $\xi$ (and $J$). Letting $\xi\searrow 0$ yields that $\V^\kappa(J)<\omega_\kappa C_3 C_2^{\kappa}<\infty$, as claimed.
\end{pr}

Note that it always possible to choose the same constant $C$ (and $\delta$) for a family of cylindrical neighborhoods and 
so if we consider the volume measure $\vol^g$ on a continuous spacetime $(M,g$) one sees that letting the 
cylindrical neighborhoods shrink and $C\searrow 1$ yields $\log_{1+2\lambda}(L)\searrow n=\dim(M)$, as the local 
doubling constant $L$ converges monotone non-increasingly to $(1+2\lambda)^n$. In the following section we can relate 
the dimensional parameter $N$ of the \emph{timelike curvature-dimension}-condition ($\mathsf{TCD}$-condition) of Cavalletti 
and Mondino \cite{CM:20} to the geometric dimension.
\bigskip

Finally, let us remark that there seems to be no obvious way to define doubling in general \LpLSn s that has all the required properties. This will be a topic of further investigation, in particular in conjunction with the relation to synthetic timelike Ricci curvature bounds as discussed in the following Section \ref{sec-ric}.

\section{Timelike Ricci curvature bounds}\label{sec-ric}
Following \cite{McC:20, MS:21}, Cavalletti and Mondino introduced synthetic lower timelike Ricci curvature bounds in \cite{CM:20} in the form of a \emph{timelike-curvature dimension condition}. In this section we show that a continuous spacetime which satisfies such a condition with respect to its volume measure, then this dimensional parameter bounds the geometric dimension from above. First, we start by briefly reviewing the relevant notions and results of \cite{CM:20}. However, for details we refer to \cite{CM:20}, especially for the definition of the (weak) timelike curvature dimension condition $\mathsf{(w)TCD^e_p}(K,N)$ and timelike measure contradiction property $\mathsf{TMCP^e_p}(K,N)$, see Section 3 of \cite{CM:20}. A measured \LpLS \Xllm is a \LpLS equipped with a Radon measure $m$ with $\supp(m)=X$. Moreover, Cavalletti-Mondino additionally require that $(X,d)$ is proper (i.e., all closed and bounded subsets are compact), making $(X,d)$ a \emph{Polish} metric space. Then synthetic timelike Ricci curvature bounds are introduced via convexity properties of entropies along curves of (causal) probability measures, cf. \cite[Def.\ 3.2, Def.\ 3.7]{CM:20}. Cavalletti-Mondino derive a timelike Brunn-Minkowski inequality for weak $\mathsf{TCD_p^e}(K,N)$ spaces \cite[Prop.\ 3.4]{CM:20} and this is the basis for the timelike Bishop-Gromov inequality \cite[Prop.\ 3.5]{CM:20}, which we will describe in more detail below.
\medskip

For $K\in\R$ define $\mathfrak{s}_{K}\colon\R\rightarrow\R$ as follows
\begin{align}
 \mathfrak{s}_{K}(t):= \begin{cases}
                        \frac{1}{\sqrt{K}} \sin(\sqrt{K} t)\quad\qquad &(K>0)\,,\\
                        t &(K=0)\,,\\
                        \frac{1}{\sqrt{-K}}\sinh(\sqrt{-K} t) &(K<0)\,.
                       \end{cases}
\end{align}

Let \Xll be a \LpLS and let $x_0\in X$, $r>0$. The \emph{sub-level set of $\tau(x_0,.)$} of radius $r$ with base point $x_0$ (or \emph{$\tau$-ball} in the terminology of \cite{CM:20}) is defined as
\begin{equation}
 B^\tau_r(x_0):=\{x\in I^+(x_0)\cup \{x_0\}: \tau(x_0,x)<r\}\,.
\end{equation}
A subset $E\subseteq I^+(x_0)\cup\{x_0\}$ is called \emph{$\tau$-star-shaped} with respect to $x_0$ if every geodesic (i.e., maximal causal curve) from $x_0$ to $x$ is contained in $E$ (except possibly at $x_0$) for all $x\in E$. Moreover we set $E_r:=\overline{B^\tau_r(x_0)}\cap E$ for $r>0$.

Then, the timelike Bishop-Gromov inequality \cite[Prop.\ 3.5]{CM:20} gives for a locally causally closed, globally hyperbolic measured \LpLS \Xllm that
\begin{equation}\label{eq-tl-bg-ine}
 \frac{m(E_r)}{m(E_R)} \geq \frac{\int_0^r \mathfrak{s}_{K/N}^N}{\int_0^R \mathfrak{s}_{K/N}^N}\,,
\end{equation}
for all $x_0\in X$, for all compact subsets $E\subseteq I^+(x_0)\cup\{x_0\}$ that are $\tau$-star-shaped with respect to $x_0$ and for all $0<r< R\leq \pi\, \sqrt{\frac{N}{\max(K,0)}}$.

Finally, to have a finite upper bound on the maximal time separation in case $K<0$, we set $R_*=\infty$, if $K\geq 0$, and for $K<0$, let $R_*\in (0,\infty)$ arbitrary but fixed for this section.

\bigskip

The following result can be understood as a kind of a (local) doubling property of the reference measure analogous to the metric case, cf.\ \cite[Cor.\ 2.4]{Stu:06b} or \cite[Cor.\ 30.14]{Vil:09}. However, it does not seem to imply a doubling property for causal diamonds --- even when restricted to continuous spacetimes.

\begin{lem}[$\mathsf{wTCD}$ implies doubling of the sub-level sets of $\tau$]\label{L:Bishop-Gromov}
 Let \Xllm be a globally hyperbolic, locally causally closed measured \LLS satisfying $\mathsf{wTCD^e_p}(K,N)$ for some $K\in\R$, $N\in[1,\infty)$ and $\mathsf{p}\in(0,1)$. Then  for all $x_0\in X$, $E\subseteq I^+(x_0)\cup\{x_0\}$ compact and $\tau$-star-shaped with respect to $x_0$ and $0<r<\min(\pi\sqrt{\frac{N}{\max(0,K)}},R_*)$ one has
 \begin{equation}\label{eq-cm-dou-con}
  m(E_{2r})\leq L\, m(E_r)\,,
 \end{equation}\label{eq-tcd-dou-con}
 where
 \begin{equation}
  L=2^{N+1}\max(1,\sgn(-K)\cosh\Bigl(\sqrt{\frac{|K|}{N}}R_*\Bigr)^N)\,.
 \end{equation}
\end{lem}
\begin{pr}
 By the timelike Bishop-Gromov inequality \eqref{eq-tl-bg-ine} one has 
 \begin{align}\label{eq-tcd-dou}
  \frac{m(E_{2r})}{m(E_{r})}\leq \frac{\int_0^{2r} \mathfrak{s}_{K/N}^N}{\int_0^{r} \mathfrak{s}_{K/N}^N} = \frac{2 \int_0^r \mathfrak{s}_{K/N}(2t)^N\,\mathrm{d}t}{\int_0^{r} \mathfrak{s}_{K/N}^N} \,.
 \end{align}
Double angle formulas show that for $K\geq 0$ the right-hand-side of Equation \eqref{eq-tcd-dou} is bounded above by $2^{N+1}$, while for $K<0$ it is bounded above by $2^{N+1} \cosh(\sqrt{\frac{-K}{N}}R_*)^N$.
\end{pr}

In the remainder of this section we establish a consistency result for continuous spacetimes that satisfy a $\mathsf{wTCD^e_p}(K,N)$ curvature bound with respect to the volume measure and a sufficient criterion for general reference measures.
\bigskip

The following theorem is a generalization of \cite[Cor.\ A.2,(2)]{CM:20} to non-smooth spacetimes.
\begin{thm}[Synthetic vs geometric dimension]\label{T:synthetic vs geometric dimension}
 Let $(M,g)$ be a continuous, globally hyperbolic and causally plain spacetime whose induced measured \LLS $(M,d^h,\vol^g,\ll,\leq,\tau)$ satisfies the $\mathsf{wTCD^e_p}(K,N)$ condition for some $K\in\R$, $N\in[1,\infty)$ and ${\mathsf p}\in(0,1)$.
 Then
 \begin{align}\label{synthetic vs geometric dimension}
  n=\dims(M)\leq N+1\,.
 \end{align}
\end{thm}
\begin{pr}
Let $p\in M$, $E\subseteq I^+(p)\cup\{p\}$ compact and $\tau$-star-shaped with respect to $p$. By the timelike Bishop-Gromov inequality \eqref{eq-tl-bg-ine} the map $r\mapsto \frac{\vol^g(E_r)}{\int_0^{r} \mathfrak{s}_{K/N}^N}$ is monotonically decreasing. As $\mathfrak{s}_{K}(t)=O(t)$ for small $t$ it suffices to show that $\vol^g(E_r)\leq c\, r^n$ for some constant $c>0$ and small $r>0$. In fact, it suffices to show that $\vol^g(B^g_r(p,U))\leq c\, r^n$, where $U$ is a causally convex neighborhood of $p$ and $B^g_r(p,U):=\{y\in U: \tau(p,y)<r\}$.

Let $(g_k)_k$ be a sequence of smooth metrics given by Lemma \ref{lem-cg-reg}. Moreover, denote by $\tau_k$  the time separation function with respect to $g_k$ ($k\in\N)$, then, in particular, by Proposition \ref{prop-tau-uni-con} we have that $\tau_k\searrow \tau$ locally uniformly and $g_k\to g$ locally uniformly as well. At this point let $U$ be a $g_0$-causally convex, relatively compact chart around $p$. Then $U$ is causally convex for $g$, $g_k$ ($k\in\N$) as $g\prec g_k\prec g_0$ for all $k\in\N$. There is a constant $c>0$ only depending on $U$ such that $\vol^{g_k}(B_r^{{g}_k}(p,U))\leq c\, r^n$, cf.\ \cite[Eq.\ (A.6)]{CM:20}.

Let $\eps>0$ and let $k\in \N$ such that $\sup_{\bar U} |\sqrt{\det{|g|}}-\sqrt{\det{|g_k|}}|< \frac{\eps}{\mathcal{H}^n(\bar U)}$. Consequently, as $\tau\leq \tau_k$, we estimate that
\begin{align}
 \vol^g(B^g_r(p,U)) = \int_U \sqrt{|\det{g}|}\, \mathbbm{1}_{B_r^g(p,U)} \leq \eps + \vol^g(B_r^{g_k}(p,U)) \leq \eps + c\, r^n\,,
\end{align}
where $\mathbbm{1}_A$ is the indicator function of a set $A$. As this holds for all $\eps>0$ we conclude $\vol^g(B^g_r(p,U))\leq c\, r^n$, as claimed.
\end{pr}

Using the same argument we can improve the bound to $n\leq N$ if we assume slightly stronger properties of $(M,g)$. In fact, the refined timelike Bishop-Gromov inequality \cite[Cor.\ 5.14]{CM:20} directly gives:
\begin{cor}[Synthetic dominates geometric dimension]\label{cor-n-leq-N}
  Let $(M,g)$ be a continuous, globally hyperbolic, timelike non-branching and causally plain spacetime of dimension $n\geq 2$ whose induced measured \LLS $(M,d^h,\vol^g,\ll,\leq,\tau)$ satisfies the $\mathsf{TMCP^e_p}(K,N)$ condition for some $K\in\R$, $N\in[1,\infty)$ and $\mathsf{p}\in(0,1)$. Moreover, assume that the causally-reversed structure satisfies the same conditions. Then 
  \begin{equation}\label{synthetic dominates geometric dimension}
  n=\dims(M)\leq N\,.
 \end{equation}
\end{cor}

By analogy with the case of metric measure geometry with positive signature \cite[Eq.~(2.4)--(2.5)]{Stu:06a} \cite{Stu:06b} 
\cite{EKS:15} \cite{CM:21},  it is natural to expect the refined Bishop-Gromov inequality (with exponent $N-1$ replacing 
$N$) to also hold under the 
$\mathsf{wTCD^e_p}(K,N)$ condition of Lemma \ref{L:Bishop-Gromov}, improving the conclusion of Theorem \ref{T:synthetic vs 
geometric dimension}  from \eqref{synthetic vs geometric dimension} to \eqref{synthetic dominates geometric dimension} 
provided the anticipated equivalence of various $\mathsf{TCD}$-notions holds true, as is known in the $\mathsf{CD}$ case.
\bigskip

\section*{Acknowledgment}
We would like to thank Christian Ketterer for helpful and stimulating discussions.

\appendix
\section{Appendix}\label{app}

Here in this appendix we establish a refined version of the \emph{Chru{\'s}ciel - Grant approximation} for continuous metrics (cf.\ \cite[Prop.\ 1.2]{CG:12}).

\begin{lem}[Smooth approximation of continuous metrics]
\label{lem-cg-reg}
 Let $(M,g)$ be a continuous spacetime. Then there is a sequence of smooth metrics $(g_k)_k$ such that
 \begin{enumerate}
  \item $g_k\to g$ locally uniformly,
  \item $g\prec g_{k+1}\prec g_k$ for all $k\in\N$, (i.e.,~$g_{k}$ has strictly wider light cones that $g_{k+1}$)
  \item $-g(X,X)< -g_k(X,X)$ for all $g$-causal $X\in TM$, and
  \item $-g_{k+1}(X,X)\leq -g_k(X,X)$ for all $g_{k+1}$-causal $X\in TM$.
 \end{enumerate}
 Note that (iii) and (iv) imply that for all $X\in TM$ $g$-causal we have that $-g(X,X)\leq -g_{k+1}(X,X)\leq -g_k(X,X)$ for all $k\in\N$.
\end{lem}
\begin{pr}
 We know that there is a net $(g_\eps)_\eps$ of smooth metrics that satisfies (analogous) points (i)-(ii), see e.g.\ \cite[Prop.\ 1.2]{CG:12} or \cite[Prop.\ 2.3(iii)]{KSV:15}. Also, one sees that the construction actually yields property (iii). Thus let $(g_\eps)_\eps$ be a net of smooth metrics that satisfies (i)-(iii).
 
Let $K\subseteq M$ be compact. Fix $0<\eps_0$ and set $L:=\{X\in TM\rvert_K: |X|_h=1,\, g(X,X)\leq 0\}$, then $L$ is compact. By property (iii) we thus have that $\delta:=\min_{X\in L}(g(X,X)-g_{\eps_0}(X,X))>0$. Let $0<\eps'<\eps_0$ such that for all $0<\eps\leq\eps'$ we have $d_h(g,g_\eps)<\delta$. Let $X\in TM\rvert_K$ be $g$-causal with $|X|_h=1$, hence $X\in L$. Consequently, we have for all $0<\eps\leq\eps'$ that
 \begin{align}\label{eq-lem-cg-reg}
  -g(X,X)&<-g_\eps(X,X) 
  \\ &\leq d_h(g,g_\eps) - g(X,X)\\
  &< \delta - g(X,X) 
  \\&\leq -g_{\eps_0}(X,X)\,.
 \end{align}
Now set $L':=\{X\in TM\rvert_K: |X|_h=1,\, g_{\eps'}(X,X)\leq 0\}$, then $L'$ is compact and consists of $g_{\eps_0}$-timelike vectors. Thus $\delta':=\min_{X\in L'}(-g_{\eps_0}(X,X))>0$. Let $0<\bar \eps\leq\eps' $ such that for all $0<\eps\leq\bar \eps$ we have $d_h(g,g_\eps)<\delta'$. Let $X\in TM\rvert_K$ be $g_\eps$-causal with $|X|_h=1$, hence $X\in L'$. Consequently, in case $X$ is not $g$-causal we have
 \begin{align}
  -g_\eps(X,X) \leq d_h(g,g_\eps) - g(X,X) < \delta' - g(X,X) < \delta' \leq -g_{\eps_0}(X,X)\,.
 \end{align}
If $X$ is $g$-causal we anyway have $-g_\eps(X,X)\leq -g_{\eps_0}(X,X)$ by Equation \eqref{eq-lem-cg-reg} above. Setting $\eps_1:=\bar\eps$, we conclude that $(g_{\eps})_{0<\eps\leq \eps_1}$ is a net that satisfies (i)-(iii) and $-g_\eps(X,X) \leq -g_{\eps_0}(X,X)$ for all $X\in TM\rvert_K$ that are $g_\eps$-causal. In particular, for all $K\subseteq M$ compact there is such a net with that properties. Applying, \cite[Lemma 2.4]{KSSV:14} to the map $(\eps,p)\mapsto g_\eps(p)$ for $0<\eps\leq \eps_1$ and $p\in M$ we globalize to get a net $(\tilde g_\eps)_\eps$ that satisfies (i)-(iii) and $-\tilde g_\eps(X,X) \leq -\tilde g_{\eps_0}(X,X)$ for all $X\in TM$ that are $\tilde g_\eps$-causal. Then start  with $\tilde g_{\eps_1}$ and continue iteratively to obtain the desired sequence.
\end{pr}

\begin{prop}[Monotonicity of time separation along approximation]
\label{prop-tau-uni-con}
 Let $(M,g)$ be a continuous, strongly causal and causally plain spacetime. Then $\tau$ is locally the uniform and nonincreasing limit of time separation functions of smooth metrics approximating $g$ but with wider light cones.
\end{prop}
\begin{pr}
Let $U$, $W$ as in Lemma \ref{lem-adv-cyl-nhd} with $\overline{U}\subseteq W$ and $W$ globally hyperbolic. As $(M,g)$ is 
causally plain the time separation function $\tau$ is lower semicontinuous by \cite[Prop.\ 5.7]{KS:18}. Thus $\tau$ is 
continuous on $\overline{U}$ by \cite[Thm.\ 3.28]{KS:18}. Moreover, let $(g_k)_k$ be a sequence as in Lemma \ref{lem-cg-reg} 
and denote the time separation function of $g_k$ by $\tau_k$ and the corresponding length functional by $L_k$. First, by the 
properties of $(g_k)_k$ it is clear that for all $p,q\in \overline{U}$ we have $\tau(p,q)\leq \tau_{k+1}(p,q)\leq 
\tau_k(p,q)$. Second, for all $p,q\in \overline{U}$ we have that $\tau_k(p,q)\to\tau(p,q)$ pointwise. To see this note that 
if $p$ is not causally related to $q$ but for every $k$ there is a $g_k$-causal curve $\gamma^k$ from $p$ to $q$ then by the 
appropriate version of the Limit Curve Theorem (\cite[Thm.\ 1.5]{Sae:16}) there is subsequence of $(\gamma^k)_k$ that 
converges uniformly to a $g$-causal curve from $p$ to $q$ (as necessarily $p\neq q$). Thus we only need to consider the case 
where $p\leq q$. To this end and similar to above, let $(\gamma^k)_k$ be a sequence of future directed $g_k$-causal and 
maximal curves from $p$ to $q$ that converge uniformly to a $g$-causal curve $\lambda$ from $p$ to $q$. Now, let $\eps>0$ and 
let $k_0\in\N$ such that for all $k\geq k_0$ we have that $|L_{k}(\lambda)-L^g(\lambda)|<\frac{\eps}{2}$. As $L_{k_0}$ is 
upper semicontinuous (cf.\ e.g.\ \cite[Thm.\ 6.3]{Sae:16}) there is a $k_1\geq k_0$ such that for all $k\geq k_1$ we have 
that $L_{k_0}(\gamma^k)\leq L_{k_0}(\lambda) + \frac{\eps}{2}$. Then for all $k\geq k_1$ we obtain
\begin{align}
 \tau_k(p,q)=L_k(\gamma^k)\leq L_{k_0}(\gamma^k)\leq L_{k_0}(\lambda) + \frac{\eps}{2} \leq L^g(\lambda) + \eps \leq 
\tau(p,q)+\eps\,.
\end{align}
At this point we can apply Dini's Theorem (cf.\ e.g.\ \cite[2.66]{AB:06}) to conclude that $\tau_k\to \tau$ uniformly on $\overline{U}$.
\end{pr}

\bibliographystyle{alphaabbr}
\bibliography{lms}

\begin{thebibliography}{GGKS18}

\bibitem[AB06]{AB:06}
C.~D. Aliprantis and K.~C. Border.
\newblock {\em Infinite dimensional analysis}.
\newblock Springer, Berlin, {T}hird edition, 2006.
\newblock A hitchhiker's guide.

\bibitem[AB21]{AB:21}
B.~Allen and A.~Burtscher.
\newblock {Properties of the Null Distance and Spacetime Convergence}.
\newblock {\em International Mathematics Research Notices}, 01 2021.

\bibitem[ACS20]{ACS:20}
L.~{Ak{\'{e}} Hau}, A.~J. {Cabrera Pacheco}, and D.~A. Solis.
\newblock On the causal hierarchy of lorentzian length spaces.
\newblock {\em Classical and Quantum Gravity}, 37(21):215013, 2020.

\bibitem[AGKS21]{AGKS:21}
S.~B. Alexander, M.~Graf, M.~Kunzinger, and C.~S\"amann.
\newblock Generalized cones as {L}orentzian length spaces: Causality,
  curvature, and singularity theorems.
\newblock {\em Comm.\ Anal.\ Geom.}, 2021.
\newblock to appear, arXiv:1909.09575 [math.MG].

\bibitem[AGS14]{AGS:14}
L.~Ambrosio, N.~Gigli, and G.~Savar\'{e}.
\newblock Metric measure spaces with {R}iemannian {R}icci curvature bounded
  from below.
\newblock {\em Duke Math. J.}, 163(7):1405--1490, 2014.

\bibitem[AT04]{AT:04}
L.~Ambrosio and P.~Tilli.
\newblock {\em Topics on analysis in metric spaces}, volume~25 of {\em Oxford
  Lecture Series in Mathematics and its Applications}.
\newblock Oxford University Press, Oxford, 2004.

\bibitem[BBI01]{BBI:01}
D.~Burago, Y.~Burago, and S.~Ivanov.
\newblock {\em A course in metric geometry}, volume~33 of {\em Graduate Studies
  in Mathematics}.
\newblock American Mathematical Society, Providence, RI, 2001.

\bibitem[BGHZ21]{BGHZ:21}
C.~Brena, N.~Gigli, S.~Honda, and X.~Zhu.
\newblock Weakly non-collapsed rcd spaces are strongly non-collapsed.
\newblock {\em Preprint, arXiv:2110.02420 [math.DG]}, 2021.

\bibitem[BH99]{BH:99}
M.~R. Bridson and A.~Haefliger.
\newblock {\em Metric spaces of non-positive curvature}, volume 319 of {\em
  Grundlehren der Mathematischen Wissenschaften [Fundamental Principles of
  Mathematical Sciences]}.
\newblock Springer-Verlag, Berlin, 1999.

\bibitem[BLMS87]{BLMS:87}
L.~Bombelli, J.~Lee, D.~Meyer, and R.~D. Sorkin.
\newblock Space-time as a causal set.
\newblock {\em Phys. Rev. Lett.}, 59(5):521--524, 1987.

\bibitem[BS18]{BS:18}
P.~Bernard and S.~Suhr.
\newblock Lyapounov {F}unctions of {C}losed {C}one {F}ields: {F}rom {C}onley
  {T}heory to {T}ime {F}unctions.
\newblock {\em Comm. Math. Phys.}, 359(2):467--498, 2018.

\bibitem[BS20]{BS:20}
E.~Bru\'{e} and D.~Semola.
\newblock Constancy of the dimension for {${\rm RCD}(K,N)$} spaces via
  regularity of {L}agrangian flows.
\newblock {\em Comm. Pure Appl. Math.}, 73(6):1141--1204, 2020.

\bibitem[Bus67]{Bus:67}
H.~Busemann.
\newblock Timelike spaces.
\newblock {\em Dissertationes Math. Rozprawy Mat.}, 53:52, 1967.

\bibitem[CC97]{CC:97}
J.~Cheeger and T.~H. Colding.
\newblock On the structure of spaces with {R}icci curvature bounded below. {I}.
\newblock {\em J. Differential Geom.}, 46(3):406--480, 1997.

\bibitem[CC00a]{CC:00a}
J.~Cheeger and T.~H. Colding.
\newblock On the structure of spaces with {R}icci curvature bounded below.
  {II}.
\newblock {\em J. Differential Geom.}, 54(1):13--35, 2000.

\bibitem[CC00b]{CC:00b}
J.~Cheeger and T.~H. Colding.
\newblock On the structure of spaces with {R}icci curvature bounded below.
  {III}.
\newblock {\em J. Differential Geom.}, 54(1):37--74, 2000.

\bibitem[CG12]{CG:12}
P.~T. Chru{\'s}ciel and J.~D.~E. Grant.
\newblock On {L}orentzian causality with continuous metrics.
\newblock {\em Classical Quantum Gravity}, 29(14):145001, 32, 2012.

\bibitem[CJN21]{CJN:21}
J.~Cheeger, W.~Jiang, and A.~Naber.
\newblock Rectifiability of singular sets of noncollapsed limit spaces with
  {R}icci curvature bounded below.
\newblock {\em Ann. of Math. (2)}, 193(2):407--538, 2021.

\bibitem[CKN82]{CaffarelliKohnNirenberg82}
L.~Caffarelli, R.~Kohn, and L.~Nirenberg.
\newblock Partial regularity of suitable weak solutions of the
  {N}avier-{S}tokes equations.
\newblock {\em Comm. Pure Appl. Math.}, 35(6):771--831, 1982.

\bibitem[CM20]{CM:20}
F.~Cavalletti and A.~Mondino.
\newblock Optimal transport in {L}orentzian synthetic spaces, synthetic
  timelike {R}icci curvature lower bounds and applications.
\newblock {\em preprint, arXiv:2004.08934 [math.MG]}, 2020.

\bibitem[CM21]{CM:21}
F.~Cavalletti and E.~Milman.
\newblock The globalization theorem for the {C}urvature-{D}imension condition.
\newblock {\em Invent. Math.}, 226(1):1--137, 2021.

\bibitem[Den21]{Deng21}
Q.~Deng.
\newblock {\em H{\"o}lder continuity of tangent cones and non-branching in
  $\mathsf{RCD}(K,N)$ spaces}.
\newblock PhD thesis, University of Toronto, 2021.

\bibitem[DL17]{DL:17}
M.~Dafermos and J.~Luk.
\newblock The interior of dynamical vacuum black holes {I}: The
  ${C^0}$-stability of the {K}err {C}auchy horizon.
\newblock {\em Preprint, arXiv:1710.01722 [gr-qc]}, 2017.

\bibitem[DPG18]{dePG:18}
G.~De~Philippis and N.~Gigli.
\newblock Non-collapsed spaces with {R}icci curvature bounded from below.
\newblock {\em J. \'{E}c. polytech. Math.}, 5:613--650, 2018.

\bibitem[EKS15]{EKS:15}
M.~Erbar, K.~Kuwada, and K.-T. Sturm.
\newblock On the equivalence of the entropic curvature-dimension condition and
  {B}ochner's inequality on metric measure spaces.
\newblock {\em Invent. Math.}, 201(3):993--1071, 2015.

\bibitem[Els18]{Els:18}
J.~Elstrodt.
\newblock {\em Ma\ss - und {I}ntegrationstheorie}.
\newblock Springer-Lehrbuch. [Springer Textbook]. Springer-Verlag, Berlin,
  {E}igth edition, 2018.
\newblock Grundwissen Mathematik. [Basic Knowledge in Mathematics].

\bibitem[Fed69]{Fed:69}
H.~Federer.
\newblock {\em Geometric measure theory}.
\newblock Die Grundlehren der mathematischen Wissenschaften, Band 153.
  Springer-Verlag New York Inc., New York, 1969.

\bibitem[Fin18]{Fin:18}
F.~Finster.
\newblock Causal {F}ermion {S}ystems: A {P}rimer for {L}orentzian {G}eometers.
\newblock {\em Journal of Physics: Conference Series}, 968(1):012004, 2018.

\bibitem[FS12]{FS:12}
A.~Fathi and A.~Siconolfi.
\newblock On smooth time functions.
\newblock {\em Math. Proc. Cambridge Philos. Soc.}, 152(2):303--339, 2012.

\bibitem[GGKS18]{GGKS:18}
M.~Graf, J.~D.~E. Grant, M.~Kunzinger, and R.~Steinbauer.
\newblock The {H}awking--{P}enrose {S}ingularity {T}heorem for
  {$C^{1,1}$}-{L}orentzian {M}etrics.
\newblock {\em Comm. Math. Phys.}, 360(3):1009--1042, 2018.

\bibitem[GJ15]{GhezziJean15}
R.~Ghezzi and F.~Jean.
\newblock Hausdorff volume in non equiregular sub-{R}iemannian manifolds.
\newblock {\em Nonlinear Anal.}, 126:345--377, 2015.

\bibitem[GKS19]{GKS:19}
J.~D.~E. Grant, M.~Kunzinger, and C.~S\"{a}mann.
\newblock Inextendibility of spacetimes and {L}orentzian length spaces.
\newblock {\em Ann. Global Anal. Geom.}, 55(1):133--147, 2019.

\bibitem[GKSS20]{GKSS:20}
J.~D.~E. Grant, M.~Kunzinger, C.~S\"{a}mann, and R.~Steinbauer.
\newblock The future is not always open.
\newblock {\em Lett. Math. Phys.}, 110(1):83--103, 2020.

\bibitem[GL17]{GL:17}
G.~J. Galloway and E.~Ling.
\newblock Some remarks on the {$C^0$}-(in)extendibility of spacetimes.
\newblock {\em Ann. Henri Poincar\'{e}}, 18(10):3427--3447, 2017.

\bibitem[GL18]{GL:18}
M.~Graf and E.~Ling.
\newblock Maximizers in {L}ipschitz spacetimes are either timelike or null.
\newblock {\em Classical Quantum Gravity}, 35(8):087001, 6, 2018.

\bibitem[GLS18]{GLS:18}
G.~J. Galloway, E.~Ling, and J.~Sbierski.
\newblock Timelike completeness as an obstruction to {$C^0$}-extensions.
\newblock {\em Comm. Math. Phys.}, 359(3):937--949, 2018.

\bibitem[Gra20]{Gra:20}
M.~Graf.
\newblock Singularity theorems for {$C^1$}-{L}orentzian metrics.
\newblock {\em Comm. Math. Phys.}, 378(2):1417--1450, 2020.

\bibitem[GS07]{GS:07}
G.~W. Gibbons and S.~N. Solodukhin.
\newblock The geometry of small causal diamonds.
\newblock {\em Phys. Lett. B}, 649(4):317--324, 2007.

\bibitem[HKM76]{HKM:76}
S.~W. Hawking, A.~R. King, and P.~J. McCarthy.
\newblock A new topology for curved space-time which incorporates the causal,
  differential, and conformal structures.
\newblock {\em J. Mathematical Phys.}, 17(2):174--181, 1976.

\bibitem[Ise15]{Ise:15}
J.~Isenberg.
\newblock On strong cosmic censorship.
\newblock In {\em Surveys in differential geometry 2015. {O}ne hundred years of
  general relativity}, volume~20 of {\em Surv. Differ. Geom.}, pages 17--36.
  Int. Press, Boston, MA, 2015.

\bibitem[KP67]{KP:67}
E.~H. Kronheimer and R.~Penrose.
\newblock On the structure of causal spaces.
\newblock {\em Proc. Cambridge Philos. Soc.}, 63:481--501, 1967.

\bibitem[KS18]{KS:18}
M.~Kunzinger and C.~S\"amann.
\newblock Lorentzian length spaces.
\newblock {\em Ann.\ Glob.\ Anal.\ Geom.}, 54(3):399--447, 2018.

\bibitem[KS21]{KS:21}
M.~Kunzinger and R.~Steinbauer.
\newblock Null distance and convergence of {L}orentzian length spaces.
\newblock {\em Preprint, arXiv:2106.05393 [math.DG]}, 2021.

\bibitem[KSS14]{KSS:14}
M.~Kunzinger, R.~Steinbauer, and M.~Stojkovi{\'c}.
\newblock The exponential map of a {$C^{1,1}$}-metric.
\newblock {\em Differential Geom. Appl.}, 34:14--24, 2014.

\bibitem[KSSV14]{KSSV:14}
M.~Kunzinger, R.~Steinbauer, M.~Stojkovi{\'c}, and J.~A. Vickers.
\newblock A regularisation approach to causality theory for
  {$C^{1,1}$}-{L}orentzian metrics.
\newblock {\em Gen. Relativity Gravitation}, 46(8):Art. 1738, 18, 2014.

\bibitem[KSSV15]{KSSV:15}
M.~Kunzinger, R.~Steinbauer, M.~Stojkovi{\'c}, and J.~A. Vickers.
\newblock Hawking's singularity theorem for {$C^{1,1}$}-metrics.
\newblock {\em Classical Quantum Gravity}, 32(7):075012, 19, 2015.

\bibitem[KSV15]{KSV:15}
M.~Kunzinger, R.~Steinbauer, and J.~A. Vickers.
\newblock The {P}enrose singularity theorem in regularity {$C^{1,1}$}.
\newblock {\em Classical Quantum Gravity}, 32(15):155010, 12, 2015.

\bibitem[LMO21]{LMO:21}
Y.~Lu, E.~Minguzzi, and S.-i. Ohta.
\newblock Geometry of weighted lorentz–finsler manifolds i: singularity
  theorems.
\newblock {\em Journal of the London Mathematical Society}, 104(1):362--393,
  2021.

\bibitem[LV09]{LV:09}
J.~Lott and C.~Villani.
\newblock Ricci curvature for metric-measure spaces via optimal transport.
\newblock {\em Ann. of Math. (2)}, 169(3):903--991, 2009.

\bibitem[McC20]{McC:20}
R.~McCann.
\newblock Displacement concavity of {B}oltzmann's entropy characterizes
  positive energy in general relativity.
\newblock {\em Camb. J. Math.}, 8(3):609--681, 2020.

\bibitem[Min15]{Min:15}
E.~Minguzzi.
\newblock Convex neighborhoods for {L}ipschitz connections and sprays.
\newblock {\em Monatsh. Math.}, 177(4):569--625, 2015.

\bibitem[Min19a]{Min:19a}
E.~Minguzzi.
\newblock Causality theory for closed cone structures with applications.
\newblock {\em Rev. Math. Phys.}, 31(5):1930001, 139, 2019.

\bibitem[Min19b]{Min:19b}
E.~Minguzzi.
\newblock Lorentzian causality theory.
\newblock {\em Living Reviews in Relativity}, 22(1):3, 2019.

\bibitem[Mit85]{Mitchell75}
J.~Mitchell.
\newblock On {C}arnot-{C}arath\'{e}odory metrics.
\newblock {\em J. Differential Geom.}, 21(1):35--45, 1985.

\bibitem[MS08]{MS:08}
E.~Minguzzi and M.~S{\'a}nchez.
\newblock The causal hierarchy of spacetimes.
\newblock In {\em Recent developments in pseudo-{R}iemannian geometry}, ESI
  Lect. Math. Phys., pages 299--358. Eur. Math. Soc., Z\"urich, 2008.

\bibitem[MS19]{MS:19}
E.~Minguzzi and S.~Suhr.
\newblock Some regularity results for {L}orentz-{F}insler spaces.
\newblock {\em Ann. Global Anal. Geom.}, 56(3):597--611, 2019.

\bibitem[MS21]{MS:21}
A.~Mondino and S.~Suhr.
\newblock An optimal transport formulation of the einstein equations of general
  relativity.
\newblock {\em Journal of the European Mathematical Society (JEMS), to appear,
  arXiv:1810.13309 [math-ph]}, 2021.

\bibitem[S{\"a}m16]{Sae:16}
C.~S{\"a}mann.
\newblock Global hyperbolicity for spacetimes with continuous metrics.
\newblock {\em Ann. Henri Poincar\'e}, 17(6):1429--1455, 2016.

\bibitem[Sbi18]{Sbi:18}
J.~Sbierski.
\newblock The {$C^0$}-inextendibility of the {S}chwarzschild spacetime and the
  spacelike diameter in {L}orentzian geometry.
\newblock {\em J. Differential Geom.}, 108(2):319--378, 2018.

\bibitem[Sbi21]{Sbi:21}
J.~Sbierski.
\newblock On holonomy singularities in general relativity and the
  {$C^{0,1}_{loc}$}-inextendibility of spacetimes.
\newblock {\em Duke Mathematical Journal, to appear, arXiv:2007.12049 [gr-qc]},
  2021.

\bibitem[SS18]{SS:18}
C.~S\"{a}mann and R.~Steinbauer.
\newblock On geodesics in low regularity.
\newblock {\em Journal of Physics: Conference Series}, 968(1):012010, 2018.

\bibitem[Stu06a]{Stu:06a}
K.-T. Sturm.
\newblock On the geometry of metric measure spaces. {I}.
\newblock {\em Acta Math.}, 196(1):65--131, 2006.

\bibitem[Stu06b]{Stu:06b}
K.-T. Sturm.
\newblock On the geometry of metric measure spaces. {II}.
\newblock {\em Acta Math.}, 196(1):133--177, 2006.

\bibitem[Sur19]{Sur:19}
S.~Surya.
\newblock The causal set approach to quantum gravity.
\newblock {\em Living Reviews in Relativity}, 22(5), 2019.

\bibitem[SV16]{SV:16}
C.~Sormani and C.~Vega.
\newblock Null distance on a spacetime.
\newblock {\em Classical Quantum Gravity}, 33(8):085001, 29, 2016.

\bibitem[Vil09]{Vil:09}
C.~Villani.
\newblock {\em Optimal transport}, volume 338 of {\em Grundlehren der
  Mathematischen Wissenschaften [Fundamental Principles of Mathematical
  Sciences]}.
\newblock Springer-Verlag, Berlin, 2009.
\newblock Old and new.

\end{thebibliography}
\addcontentsline{toc}{section}{References}

\end{document}